\documentclass[journal,12pt,draftclsnofoot,onecolumn, narroweqnarray]{IEEEtran}

\ifCLASSINFOpdf
\else
\fi

\usepackage{color}
\usepackage{algorithm}
\usepackage{algorithmic}
\usepackage{graphicx}
\usepackage{epstopdf}
\usepackage{mathrsfs}
\DeclareMathAlphabet{\mathpzc}{OT1}{pzc}{m}{it}

\newcounter{cc}
\usepackage{graphicx}
\usepackage{caption}
\usepackage{subcaption}

\usepackage{cite}
\usepackage{bbm}
\usepackage{float}
\usepackage{amsmath}
\usepackage{amsfonts}
\usepackage{amssymb}
\usepackage{amsthm}
\usepackage{slashbox}
\usepackage{bbding}
\usepackage{graphicx}
\usepackage{dsfont}
\usepackage{clrscode}
\usepackage[T1]{fontenc}
\usepackage{times}
\usepackage{mathptmx}
\usepackage{enumerate}
\usepackage[T1]{fontenc}
\usepackage{bm}

\usepackage{fancyhdr}
\definecolor{awesome}{rgb}{1.0, 0.13, 0.32}

\allowdisplaybreaks

\makeatletter

\newcommand{\Rmnum}[1]{\expandafter\@slowromancap\romannumeral #1@}
\makeatother

\DeclareMathOperator{\E}{\mathbb{E}}
\DeclareMathOperator{\Tr}{Tr}
\DeclareMathOperator*{\argmin}{\arg\!\min}

\usepackage{mathtools}

\newtheorem{proposition}{Proposition}

\allowdisplaybreaks

\begin{document}

\title{ MmWave Amplify-and-Forward MIMO Relay Networks with Hybrid Precoding/Combining Design}
 \author {\IEEEauthorblockN{Lisi Jiang, Hamid Jafarkhani}
\thanks{ The authors are with the Center for Pervasive Communications and Computing, University of California, Irvine (email: \{lisi.jiang, hamidj\}@uci.edu). This work was supported in part by the NSF Award ECCS-1642536. This paper was presented in part at the IEEE International Conference on Communications 2019 \cite{jiang2019hybrid}.}\\\
 }
\maketitle
\thispagestyle{plain}
\pagestyle{plain}

\begin{abstract}
In this paper, we consider the amplify-and-forward relay networks in mmWave systems and propose a hybrid precoder/combiner design approach. The phase-only RF precoding/combining matrices are first designed to  support multi-stream transmission, where we compensate the phase for the eigenmodes of the channel. Then, the baseband precoders/combiners are performed to achieve the maximum mutual information. Based on the data processing inequality for the mutual information, we first jointly design the baseband source and relay nodes to maximize the mutual information before the destination baseband receiver. The proposed  low-complexity iterative algorithm for the source and relay nodes is based on the  equivalence between mutual information maximization and the weighted MMSE. After we obtain the optimal precoder and combiner for the source and relay nodes, we implement the MMSE-SIC filter at the baseband receiver to keep the mutual information unchanged, thus obtaining the optimal mutual information for the whole relay system. Simulation results show that our algorithm achieves better performance with lower complexity compared with other algorithms in the literature. In addition, we also propose a
robust joint transceiver design for imperfect channel state information.
 
\end{abstract}


\IEEEpeerreviewmaketitle
\vspace{-0.2cm}
\section{Introduction}
Communications over millimeter wave (mmWave) has received significant attention recently because of the high data rates provided by the large bandwidth at the mmWave carrier frequencies. Also, using large antenna arrays in mmWave communication systems is possible because the small wavelength  allows integrating many antennas in a small area. Despite its advantages, the mmWave carrier frequencies suffer from relatively severe propagation losses. Meanwhile, the sparsity of the mmWave scattering environment usually results in rank-deficient channels \cite{mmwave_book}. 

To overcome the large path losses, large antenna arrays can be placed at both transmitters and receivers to guarantee sufficient received signal power \cite{intro_mmWave}. The large antenna arrays lead to a large number of radio frequency (RF) chains, which greatly increase the implementation cost and complexity. To reduce the number of RF chains, hybrid analog/digital precoding has been proposed, which connects analog phase shifters with a reduced number of RF chains. The main advantage of the hybrid precoding is that it can trade off between the low-complexity limited-performance analog phase shifters and the high-complexity good-performance digital precoding \cite{han2015large}.

Despite the help of large antenna arrays, the severe propagation losses still limit mmWave communications to take place within short ranges. Fortunately, the coverage can be greatly extended with the help of relay nodes \cite{laneman_relay}. Therefore, investigating the performance of hybrid precoding/combining in the relay scenario is important. For the conventional relay scenario, network beamforming in  amplify-and-forward (AF) relay networks was studied in \cite{jing2008network,jing2009network}. 

For a mmWave relay scenario, large antenna arrays are usually implemented to mitigate the severe path loss. In addition, a hybrid precoding method is adopted. There are  two typical hybrid precoding structures: (i) fully-connected structure (where each RF chain is connected to all antennas) \cite{el2014spatially}, and (ii) sub-connected structure (where each RF chain is connected to a subset of antennas) \cite{gao2018low}.  For fully-connected mmWave networks with AF relay nodes, the authors in \cite{single_relay_mmwave} designed hybrid precoding matrices using the orthogonal matching pursuit (OMP) algorithm. However, the performance of the OMP algorithm used in \cite{single_relay_mmwave} depends on the orthogonality of the pre-determined candidates for the analog precoders. In \cite{xue2016joint}, a joint source and relay precoding design for mmWave AF relay network is proposed based on semidefinite programming (SDP). However, the proposed method in \cite{xue2016joint} is only suitable for one data stream scenario. In \cite{sheu2017hybrid}, to reduce the complexity, the RF and the baseband (BB) are separately designed and a minimum mean squared error (MMSE)-based design for the BB filters is proposed. Although the algorithm in \cite{sheu2017hybrid} shows its advantage over the OMP algorithm in terms of sum spectral efficiency, it did not optimize the sum rate of the system. In fact, \cite{sheu2017hybrid} can be seen as a special case of our proposed methods since we minimize the weighted mean squared error. In \cite{tsinos2018hybrid}, an efficient algorithm is proposed via employing the so-called Alternating Direction Method of Multipliers (ADMM), which greatly reduce the distance between the hybrid precoder/combiner and the full-digital precoder/combiner. However, the ADMM algorithm has a high complexity and is sub-optimal in terms of the data rate for the system. For sub-connected structures, \cite{xue2018relay} proposes a MMSE-based relay hybrid precoding design. To make the problem tractable, \cite{xue2018relay} reformulates the original problem as three subproblems and proposes an iterative successive approximation (ISA) algorithm. The algorithm in \cite{xue2018relay} can also be extended to the fully-connected structure. Compared with the OMP algorithm, the ISA algorithm in \cite{xue2018relay} greatly improves the performance, however, the complexity of the ISA algorithm is high and it only optimizes the relay node.

In this paper, we study the hybrid precoding for fully-connected mmWave AF relay networks in the domain of massive multiple-input and multiple-output (MIMO) systems. To reduce the complexity, we separate the RF and the BB.  For the RF, we first design the phase-only RF precoding/combining matrices for multi-stream transmissions. 
We decompose the channel into parallel sub-channels through singular value decomposition (SVD) and compensate the phase of each sub-channel, i.e., each eigenmode of the channel. 
When the RF precoding and combining are performed, the digital baseband precoders/combiners are performed on the equivalent baseband channel to achieve the maximal mutual information. The problem of finding the optimal baseband precoders/combiners for the optimal mutual information is non-convex and intractable to solve by low complexity methods. Based on the data processing inequality for the mutual information \cite{gray2011entropy}, we first jointly design the baseband source and relay nodes to maximize the mutual information before the destination baseband receiver. We propose a low-complexity iterative algorithm to design the precoder and combiner for the source and relay nodes, which is based on the equivalence between mutual information maximization and the weighted MMSE \cite{christensen2008weighted}. After we obtain the optimal precoder and combiner for the source and relay nodes, we implement the MMSE successive interference cancellation (MMSE-SIC) filter \cite{tse2005fundamentals} at the baseband receiver to keep the mutual information unchanged, thus obtaining the optimal mutual information for the whole relay system. Simulation results show that our algorithm outperforms the OMP in \cite{single_relay_mmwave}. Moreover, our algorithm achieves better performance with lower complexity compared to the ISA algorithm in \cite{xue2018relay}.

We also propose a robust hybrid precoding/combining approach considering the inevitable imperfect channel state information (CSI) in the second part of the paper. Robust design for traditional relay systems has been well studied in papers, such as \cite{xing2010robust,xing2010transceiver,rong2011robust}. In \cite{koyuncu2012distributed,liu2016amplify}, the topic of imperfect channel state information in amplify-and-forward relay networks has been studied under amplify-and-forward relay networks with limited feedback. However, there is not much work on the effects of  imperfect channel state information in mmWave relay networks. In \cite{luo2018robust}, a robust OMP-based algorithm is proposed to maximize the receiving signal-to-noise ratio (SNR) at the destination node. Similar with the non-robust case, the performance of the OMP-based algorithm depends on the orthogonality of the predetermined candidates for the analog precoders. In this paper, we adopt the well-known Kronecker model \cite{zhang2008statistically,xing2010robust} for the CSI mismatch. We first estimate the phase for RF precoding/combining to minimize the average estimation error. Then, we modify our proposed weighted MMSE approach for the perfect CSI to achieve a more robust performance for the baseband processing. Simulation results demonstrate the robustness of the proposed algorithm against CSI mismatch.

The contributions of our paper can be summarized as follows: 
\begin{itemize}
\item We propose a hybrid precoding/combining approach for perfect CSI in mmWave relay systems. The phase-only RF precoding/combining matrices are first designed to achieve large array gains and support multi-stream transmissions. Then, we design the baseband processing system to achieve maximal mutual information by transforming the highly complicated non-convex mutual information maximization problem into an easily tractable weighted MMSE problem. An iterative algorithm which decouples the joint design into four sub-problems is developed. 
\item A robust design for the imperfect CSI is further proposed by modifying the non-robust precoding/combining design. To the best of our knowledge, this is the first mmWave relay system design that is robust to channel estimation error. Numerical results are provided to show the robustness of the proposed algorithm against CSI mismatch.
\end{itemize}

Compared with our conference version \cite{jiang2019hybrid} which designs the hybrid filters for perfect CSI, we analyze the impact of imperfect CSI in this paper, and further propose a robust design to strengthen the robustness of our proposed algorithm. The remaining sections are organized as follows. In Section \uppercase\expandafter{\romannumeral2} , we describe the system model and the mmWave channel model. Section \uppercase\expandafter{\romannumeral3} formulates the proposed hybrid precoding/combing approach for the perfect CSI. Section \uppercase\expandafter{\romannumeral4} presents the proposed robust hybrid design for the imperfect CSI. Numerical examples are presented and discussed in Section \uppercase\expandafter{\romannumeral5}.  We provide concluding remarks in Section \uppercase\expandafter{\romannumeral6}.

\emph{Notation:} $\mathbb{C}^{m\times n}$ is the set of all $m\times n$ complex matrices with $\mathbb{C}^m \triangleq \mathbb{C}^{m\times 1}$ and $\mathbb{C}\triangleq\mathbb{C}^1$. $\mathbb{I}_m$ is the $m\times m$ identity matrix, and $\mathbf{0}_{m\times n}$ is the $m\times n$ all-zero matrix. $\mathbb{CN}(\bm{\mu},\mathbf{K})$ is a circularly-symmetric complex Gaussian random vector with mean vector $\bm{\mu}$ and covariance matrix $\mathbf{K}$. Matrices $\mathbf{A}^T$ and $\mathbf{A}^{H}$ are the transpose and the Hermite transpose of matrix $\mathbf{A}$, respectively. Matrix $\mathbf{A}=[\bm{\alpha}_1, \bm{\alpha}_2,...,\bm{\alpha}_L]$ represents the concatenation of the $L$ vectors $\bm{\alpha}_i$, and $\mathbf{B}=[\mathbf{A}_{1},\mathbf{A}_{2},...,\mathbf{A}_{K}]$ represents the concatenation of the $K$ matrices $\mathbf{A}_i$.

\section{System Model}
In this section, we present the signal and channel model for a single user  mmWave MIMO relay system with large antenna arrays and limited RF chains.
\begin{figure*}[htp]
\setlength{\abovedisplayskip}{3pt}
\setlength{\belowdisplayskip}{3pt}
\centering
\includegraphics[width=14cm]{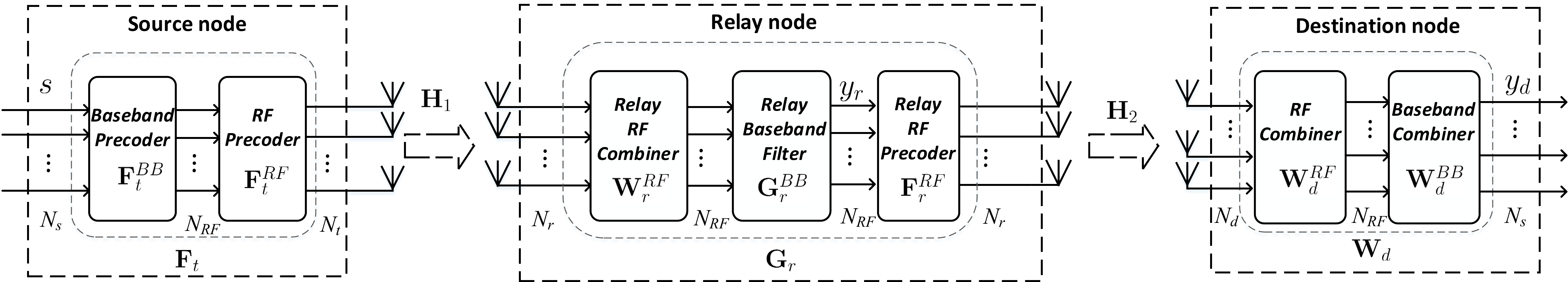}
\caption{System model}
\label{system model}
\end{figure*}
\subsection{System model}
	Consider a single-user mmWave MIMO relay system using hybrid precoding as illustrated in Fig. 1. The system consists of a source node with ${\text{N}_\text{t}}$ transmission antennas, a relay node with $\text{N}_\text{r}$ antennas for both transmitting and receiving signals, and  a destination node with $\text{N}_\text{d}$ antennas. Assuming $\text{N}_\text{s}$ data streams are transmitted, the BS is equipped with $\text{N}_\text{RF}$ RF chains such that $\text{N}_\text{s} \leq \text{N}_\text{RF} \leq \text{N}_\text{t}$. Using the {$\text{N}_\text{RF}$} transmit chains, an {$\text{N}_\text{RF} \times \text{N}_\text{s}$} baseband precoder {$\mathbf{F}_\text{t}^\text{BB}$} is applied. The RF precoder is an {$\text{N}_\text{t} \times \text{N}_\text{RF}$} matrix {$\mathbf{F}_\text{t}^\text{RF}$}.  Half duplex relaying is adopted. During the first time slot, the BS transmits the {$\text{N}_\text{s}$} data streams to the relay through a MIMO channel {$\mathbf{H}_1 \in \mathbb{C}^{\text{N}_\text{r} \times \text{N}_\text{t}}$}. The relay receives the signal with an RF combiner {$\mathbf{W}_\text{r}^\text{RF} \in \mathbb{C}^{\text{N}_\text{RF} \times \text{N}_\text{r}}$} and a baseband filter {$\mathbf{G}_\text{r}^\text{BB} \in \mathbb{C}^{\text{N}_\text{RF} \times \text{N}_\text{RF}}$}. During the second time slot, the relay transmits the data using one RF percoder {$\mathbf{F}_\text{r}^\text{RF}\in \mathbb{C}^{\text{N}_\text{r} \times \text{N}_\text{RF}}$} through a MIMO channel {$\mathbf{H}_2 \in \mathbb{C}^{\text{N}_\text{d} \times \text{N}_\text{r}}$} and the destination receives the data with one RF combiner {$\mathbf{W}_\text{d}^\text{RF} \in \mathbb{C}^{\text{N}_\text{RF} \times \text{N}_\text{d}}$} and one baseband combiner {$\mathbf{W}_\text{d}^\text{BB} \in \mathbb{C}^{\text{N}_\text{s} \times \text{N}_\text{RF}}$}. 

We assume the transmited signal is {$\bm{s}=[s^1, s^2,...,s^{\text{N}_\text{s}}]^T$}  with {$\E[\bm{s}\bm{s}^{H}]=\mathbf{I}_{\text{N}_\text{s}} \in \mathbb{C}^{\text{N}_\text{s} \times \text{N}_\text{s}}$}. During the first time slot, the received signal after the baseband filter at the relay can be expressed as
\begin{equation}
\label{1}
	\bm{y}_\text{r}=\mathbf{G}_\text{r}^\text{BB}\mathbf{W}_\text{r}^\text{RF}\mathbf{H}_1\mathbf{F}_\text{t}^\text{RF}\mathbf{F}_\text{t}^\text{BB}\bm{s}+\mathbf{G}_\text{r}^\text{BB}\mathbf{W}_\text{r}^\text{RF}\bm{n}_1,
\end{equation}
where {$ \bm{n}_1 \in \mathbb{C}^{\text{N}_\text{r} \times 1}$}  is a zero-mean complex Gaussian noise vector at the relay node with covariance matrix {$\E[\bm{n}_1\bm{n}_1^{H}]=\sigma_1^2\mathbf{I}_{\text{N}_\text{r}} \in \mathbb{C}^{\text{N}_\text{r} \times \text{N}_\text{r}}$}.
The power constraint at the source node is
\begin{equation}
\label{2}
\left\|\mathbf{F}_\text{t}^\text{RF}\mathbf{F}_\text{t}^\text{BB}\right\|^2_F\leq \text{E}_\text{t}.
\end{equation}
During the second time slot, the received signal after the combiners at the destination can be expressed as
\begin{equation}
\begin{split}
\label{4}
\bm{y}_\text{d}&=  \mathbf{W}_\text{d}^\text{BB}\mathbf{W}_\text{d}^\text{RF}\mathbf{H}_2\mathbf{F}_\text{r}^\text{RF}\mathbf{G}_\text{r}^\text{BB}\mathbf{W}_\text{r}^\text{RF}\mathbf{H}_1\mathbf{F}_\text{t}^\text{RF}\mathbf{F}_\text{t}^\text{BB}\bm{s}\\
&+\mathbf{W}_\text{d}^\text{BB}\mathbf{W}_\text{d}^\text{RF}\mathbf{H}_2\mathbf{F}_\text{r}^\text{RF}\mathbf{G}_\text{r}^\text{BB}\mathbf{W}_\text{r}^\text{RF}\bm{n}_1+\mathbf{W}_\text{d}^\text{BB}\mathbf{W}_\text{d}^\text{RF}\bm{n}_2,
\end{split}
\end{equation}
where {$ \bm{n}_2 \in \mathbb{C}^{\text{N}_\text{d} \times 1}$} is a zero-mean complex Gaussian noise vector at the destination node with covariance matrix {$\E[\bm{n}_2\bm{n}_2^{H}]=\sigma_2^2\mathbf{I}_{\text{N}_\text{d}} \in \mathbb{C}^{\text{N}_\text{d} \times \text{N}_\text{d}}$}.

To simplify the expression, we define {$\mathbf{F}_\text{t}=\mathbf{F}_\text{t}^\text{RF}\mathbf{F}_\text{t}^\text{BB} \in \mathbb{C}^{\text{N}_\text{t}\times \text{N}_\text{s}}$} as the hybrid precoding matrix at the transmitter, {$\mathbf{G}_\text{r}=\mathbf{W}_\text{r}^\text{RF}\mathbf{G}_\text{r}^\text{BB}\mathbf{F}_\text{r}^\text{BB} \in \mathbb{C}^{\text{N}_\text{r}\times \text{N}_\text{r}}$} as the hybrid filter at the relay node, and {$\mathbf{W}_\text{d}=\mathbf{W}_\text{d}^\text{BB}\mathbf{W}_\text{d}^\text{RF} \in \mathbb{C}^{\text{N}_\text{s}\times \text{N}_\text{d}}$} as the hybrid combiner at the destination node. Eq. (\ref{4}) can be expressed as
\begin{equation}
\label{4.1}
\bm{y}_\text{d}=  \mathbf{W}_\text{d}\mathbf{H}_2\mathbf{G}_\text{r}\mathbf{H}_1\mathbf{F}_\text{t}\bm{s}+\mathbf{W}_\text{d}\mathbf{H}_2\mathbf{G}_\text{r}\bm{n}_1+\mathbf{W}_\text{d}\bm{n}_2.
\end{equation}

The relay's power constraint is
\begin{equation}
\label{3}
\E[\left\|\mathbf{G}_\text{r}\mathbf{H}_1\mathbf{F}_\text{t}\bm{s}+\mathbf{G}_\text{r}\bm{n}_1\right\|^2_F]\leq \text{E}_\text{r}.
\end{equation}

Based on this hybrid precoding/combining system model, we can derive the achieved data rate for the system as
\begin{equation}
\label{rate}
R=\frac{1}{2}\log_2\det( \mathbf{I}_{\text{N}_\text{s}}+\mathbf{W}_\text{d}\mathbf{H}_2\mathbf{G}_\text{r}\mathbf{H}_1\mathbf{F}_\text{t}\mathbf{R}_\text{n}^\text{-1}\mathbf{F}_\text{t}^H\mathbf{H}_1^H\mathbf{G}_\text{r}^H\mathbf{H}_2^H\mathbf{W}_\text{d}^H),
\end{equation}where {$\mathbf{R}_\text{n}=\sigma_1^2\mathbf{W}_\text{d}\mathbf{H}_2\mathbf{G}_\text{r}\mathbf{G}_\text{r}^H\mathbf{H}_2^H\mathbf{W}_\text{d}^H+\sigma_2^2\mathbf{W}_\text{d}\mathbf{W}_\text{d}^H$} is the covariance matrix of the colored Gaussian noise at the output of the baseband combiner. 

Generally, we want to jointly optimize the RF and baseband precoders/combiners to achieve the optimal data rate. However, finding the global optima for this problem (maxmizing $R$
while imposing constant-amplitude on the RF analog precoder/combiners) is non-convex and intractable. Separated RF and baseband processing designs, as \cite{ni2016hybrid} did, are investigated to obtain satisfying performance. Therefore, we will separate the RF and baseband domain designs in this paper.
\subsection{Channel model}
MmWave channels are expected to have limited scattering characteristic \cite{rappaport2013broadband,smulders1997characterisation,xu2002spatial}, which means the assumption of a rich scattering environment becomes invalid. This is called sparsity in the literature and leads to the unreliability of traditional channel models, such as the Rayleigh fading channel model. To characterize the limited scattering feature, we adopt the clustered mmWave channel model in \cite{mmwave_book,smulders1997characterisation,xu2002spatial,sayeed2002deconstructing} with $\text{L}$ scatters. Each scatter is assumed to contribute $\text{N}_\text{cl}$ propagation paths to the channel matrix H. Then, the channel is given by
\begin{equation}
\label{5}
\mathbf{H}=\sqrt{\frac{\text{N}_\text{t}\text{N}_\text{r}}{\text{L}\text{N}_\text{cl}}}\sum_{l=1}^\text{L}\sum_{n=1}^{\text{N}_\text{cl}} \alpha_{l,n} \bm{a}_\text{r} (\varphi_{l,n}^\text{r},\theta_{l,n}^\text{r}) \bm{a}^H_\text{t} (\varphi_{l,n}^\text{t},\theta_{l,n}^\text{t}),
\end{equation}
where $\alpha_{l,n}$ is the complex
gain of the {$n^\text{th}$} path in the {$l^\text{th}$} scatter with the distribution $\mathbb{CN}(0,1)$, {$\varphi_{l,n}^\text{r}(\theta_{l,n}^\text{r})$} and {$\varphi_{l,n}^\text{t}(\theta_{l,n}^\text{t})$} are the random azimuth and elevation angles of arrival and departure. $\bm{a}_\text{r} (\varphi_{l,n}^\text{r},\theta_{l,n}^\text{r})$ and $\bm{a}_\text{t} (\varphi_{l,n}^\text{t},\theta_{l,n}^\text{t})$ are the receiving  and transmitting antenna array response vectors, respectively. While the algorithms and results in the paper can be applied to arbitrary antenna arrays, we use uniform linear arrays (ULAs) in the simulations for simplicity. The array response vectors take the following form \cite{balanis2016antenna}:
\begin{equation}
\label{6}
\bm{a}^\text{ULA}(\varphi)=\frac{1}{\sqrt{\text{N}}}[1,e^{jkdsin(\varphi)},...,e^{j(\text{N}-1)kdsin(\varphi)}]^T,
\end{equation}where $k=\frac{2\pi}{\lambda}$. Parameter $\lambda$ represents the wavelength of the carrier and $d$ is the spacing between antenna elements. The angle $\varphi$ is assumed to have a uniform distribution over $[0, 2\pi]$.

Since the channel in mmWave systems has limited scattering, we can further simplify the channel by assuming each scatter only contributes one path to the channel matrix. Then, the channel can be further expressed as 
\begin{equation}
\label{5.1}
\mathbf{H}=\sqrt{\frac{\text{N}_\text{t}\text{N}_\text{r}}{\text{L}}}\sum_{l=1}^\text{L} \alpha_{l} \bm{a}_r (\varphi_{l}^r,\theta_{l}^r) \bm{a}^H_t (\varphi_{l}^t,\theta_{l}^t).
\end{equation}

The matrix formulation can be expressed as
\begin{equation}
\label{channel_simp}
\mathbf{H}=\sqrt{\frac{\text{N}_\text{t}\text{N}_\text{r}}{\text{L}}}\mathbf{A}_\text{r}diag(\bm{\alpha})\mathbf{A}_\text{t}^H, 
\end{equation}
where {$\mathbf{A}_\text{r}=[\bm{a}_\text{r} (\varphi_1^\text{r},\theta_1^\text{r}),...,\bm{a}_\text{r} (\varphi_\text{L}^\text{r},\theta_\text{L}^\text{r})]$} and {$\mathbf{A}_\text{t}=[\bm{a}_\text{t} (\varphi_1^\text{t},\theta_1^\text{t}),...,\bm{a}_\text{t} (\varphi_\text{L}^\text{t},\theta_\text{L}^\text{t})]$} are matrices containing the receiving and transmitting array response vectors, and {$\bm{\alpha}=[\alpha_1,...,\alpha_\text{L}]^T$}.
\vspace{-0.1cm}

\section{Hybrid Precoder/Combiner Design }
As discussed in Section \uppercase\expandafter{\romannumeral 2}, we use a hybrid design to reduce the number of RF chains. We first design the RF precoder/combiner. Then, based on the designed RF precoder/combiner, we design a low-complexity iterative algorithm for the baseband precoder/combiner to maximize the mutual information.
\subsection{RF precoder/combiner design}
Our goal for RF precoder/combiner is to make the channels decomposed into {$\text{N}_\text{RF}$} parallel sub-channels to support the multi-stream transmission. The main challenge is the constant-magnitude constraints on RF precoders and combiners. Without the constant-magnitude constraints, the optimal precoder/combiner should be the right/left singular matrix of the channel, which transmits the signals along the eigenmodes of the channel. Considering the constant-magnitude constraints, we cannot directly use the singular matrix to rotate the signals, but we can use the projection on each eigenmode as a criterion to choose RF precoder and combiner. For the { $i^\text{th}$} eigenmode, the best precoder should be the one that has the largest projection on that eigenmode, i.e., the one that casts the most energy along that eigenmode direction. 

Using {$\mathbf{H}_1$} as an example, we first perform the singular value decomposition (SVD) for the channel matrix.
\begin{equation}
\label{svd}
\mathbf{H}_1 = \mathbf{U}_1\Sigma_1\mathbf{V}_1^H = \sum_{i=1}^\text{L}\sigma_i\bm{u}_i\bm{v}_i^H,
\end{equation}
where {$\bm{u}_i$} and {$\bm{v}_i$} are the { $i^\text{th}$} vectors in matrices {$\mathbf{U}_1$} and {$\mathbf{V}_1$}, respectively, which correspond to {$\sigma_i$}. The singular values $\sigma_i$s are assumed to be in a descending order. $\text{L}$ is the rank of the channel and is equal to the number of propagation paths for the mmWave scenario. Note that for mmWave systems, the channels usually have limited scattering characteristics, which means the number of propagation paths is far less than $\min(\text{N}_\text{t}, \text{N}_\text{r})$. In such cases, the channel rank is equal to the number of propagation paths $\text{L}$. Eq. (\ref{svd}) indicates that channel {$\mathbf{H}_1$} has $\text{L}$ eigenmodes. We denote the {$i^\text{th}$} eigenmode by {$\bm{u}_i\bm{v}_i^H$}, and its gain by {$\sigma_i$}.

For our RF precoding/combining, we want to maximize the projection of the {$i^\text{th}$} data stream onto the {$i^\text{th}$} eigenmode, i.e., {$\left|\bm{w}_i^H\bm{u}_i\bm{v}_i^H\bm{f}_i\right|$}, where  { $\bm{f}_i$} and {$\bm{w}_i$} are the {$i^\text{th}$} vector of precoder {$\mathbf{F}_t^{RF}$} and combiner {$\mathbf{W}_\text{r}^\text{RF}$}, respectively.
To approach the maximal projection, we have the following proposition.
\vspace{-0.2cm}
\begin{proposition}
The optimal phase-only vectors {$\bm{f}_i$} and {$\bm{w}_i$}, which maximize the projection for the {$i^\text{th}$} data stream onto the {$i^\text{th}$} eigenmode of the channel, will satisfy the following conditions:
\begin{equation}
\label{c1}
phase(\bm{f}_i[m])=phase(\bm{v}_i[m])~\forall m=1,2,...,\text{N}_\text{t},
\end{equation}
\begin{equation}
\label{c2}
phase(\bm{w}_i[n])=phase(\bm{u}_i[n])~\forall n=1,2,...,\text{N}_\text{r},
\end{equation}
where {$\bm{\cdot}[k]$} represents the {$k^\text{th}$} element of a vector.
\end{proposition}
\begin{proof}
First, we express the vectors in polar coordinates. Due to the magnitude-constant constraints, vectors {$\bm{f}_i$} and {$\bm{w}_i$} are expressed as {$\bm{f}_i=\frac{1}{\sqrt{\text{N}_\text{t}}}[e^{j\theta_1^i}, e^{j\theta_2^i},...,e^{j\theta_{\text{N}_\text{t}}^i}]^T$} and {$\bm{w}_i=\frac{1}{\sqrt{\text{N}_\text{r}}}[e^{j\varphi_1^i}, e^{j\varphi_2^i},...,e^{j\varphi_{\text{N}_\text{r}}^i}]^T$}. Since there are no constant-magnitude constraints for  {$\bm{v}_i$} and {$\bm{u}_i$}, each element in the vector has its own magnitude. The polar forms of {$\bm{v}_i$} and {$\bm{u}_i$} are {$\bm{v}_i=[r_1^ie^{j\alpha_1^i},r_2^ie^{j\alpha_2^i},...,r_{\text{N}_\text{t}}^ie^{j\alpha_{\text{N}_\text{t}}^i}]^T$} and {$\bm{u}_i=[\rho_1^ie^{j\beta_1^i},\rho_2^ie^{j\beta_2^i},...,\rho_{\text{N}_\text{r}}^ie^{j\beta_{\text{N}_\text{r}}^i}]^T$}, respectively. 
Then, the projection can be calculated as
\begin{equation}
\left|\bm{w}_i^H\bm{u}_i\bm{v}_i^H\bm{f}_i\right| = 
\left|\frac{1}{\sqrt{\text{N}_\text{r}}}\sum_{n=1}^{\text{N}_\text{r}}\rho_n^ie^{j(\varphi_n^i-\beta_n^i)}\right| \left|\frac{1}{\sqrt{\text{N}_\text{t}}}\sum_{m=1}^{\text{N}_\text{t}}r_m^ie^{j(\alpha_m^i-\theta_m^i)}\right|.
\end{equation}

According to the Cauchy-Schwartz inequality, we have
\begin{equation}
\begin{split}
\label{ss1}
&\left|\frac{1}{\sqrt{\text{N}_\text{t}}}\sum_{m=1}^{\text{N}_\text{t}}r_m^ie^{j(\alpha_m^i-\theta_m^i)}\right|^2\\
&\leq\frac{1}{\text{N}_\text{t}}\sum_{m=1}^{\text{N}_\text{t}}\left|r_m^i\right|^2\sum_{m=1}^{\text{N}_\text{t}}\left|e^{j(\alpha_m^i-\theta_m^i)}\right|^2=\frac{1}{\text{N}_\text{t}}\sum_{m=1}^{\text{N}_\text{t}}\left|r_m^i\right|^2.
\end{split}
\end{equation}
Equality can be achieved in (\ref{ss1}) if and only if {$\theta_m^i=\alpha_m^i~ \forall m=1,2,...,\text{N}_\text{t}$}. This means the maximal {$\left|\bm{v}_i^H\bm{f}_i\right|$} is achieved when {$\theta_m^i=\alpha_m^i~ \forall m=1,2,...,\text{N}_\text{t}$}.
Similarly, the maximal {$\left|\bm{w}_i^H\bm{u}_i\right|$} is achieved when {$\varphi_n^i=\beta_n^i ~ \forall n=1,2,...,\text{N}_\text{r}$}. Therefore, we have the conclusion that the optimal phase-only vectors {$\bm{f}_i$} and {$\bm{w}_i$}, which maximize {$|\bm{w}_i^H\bm{u}_i\bm{v}_i^H\bm{f}_i|$}, will satisfy the conditions in (\ref{c1}) and (\ref{c2}).
%
%
\end{proof}
Our RF precoders and combiners are actually compensating the phase of each sub-channel. Note that when the number of antennas is large enough, the response vectors {$\bm{a}_\text{r} (\varphi_l^\text{r},\theta_l^\text{r})$s and $\bm{a}_\text{t} (\varphi_l^\text{t},\theta_l^\text{t})$s} become orthogonal to each other.  {$\mathbf{A}_\text{t}$} and {$\mathbf{A}_\text{r}$} become the left and right singular matrices of the channel and they directly become our RF precoder and combiner. In this case, we can perfectly decompose the channel into independent parallel sub-channels. The equivalent channel after RF processing is diagonal, which makes it easier for baseband processing.

\subsection{Baseband system}
In this section, we focus on designing the baseband precoding/combining matrices. First, we define the equivalent baseband channels for {$\mathbf{H}_1$} and {$\mathbf{H}_2$} as
\begin{equation}
\label{effechannel1}
\tilde{\mathbf{H}}_1= \mathbf{W}_\text{r}^\text{RF}\mathbf{H}_1\mathbf{F}_\text{t}^\text{RF},
\end{equation}
\begin{equation}
\label{effechannel2}
\tilde{\mathbf{H}}_2= \mathbf{W}_\text{d}^\text{RF}\mathbf{H}_2\mathbf{F}_\text{r}^\text{RF}.
\end{equation}

Based on the equivalent channels, we simplify our system model as shown in Fig. \ref{baseband system}.
\begin{figure*}[htp]
\setlength{\abovedisplayskip}{3pt}
\setlength{\belowdisplayskip}{3pt}
\centering
\includegraphics[width=8 cm]{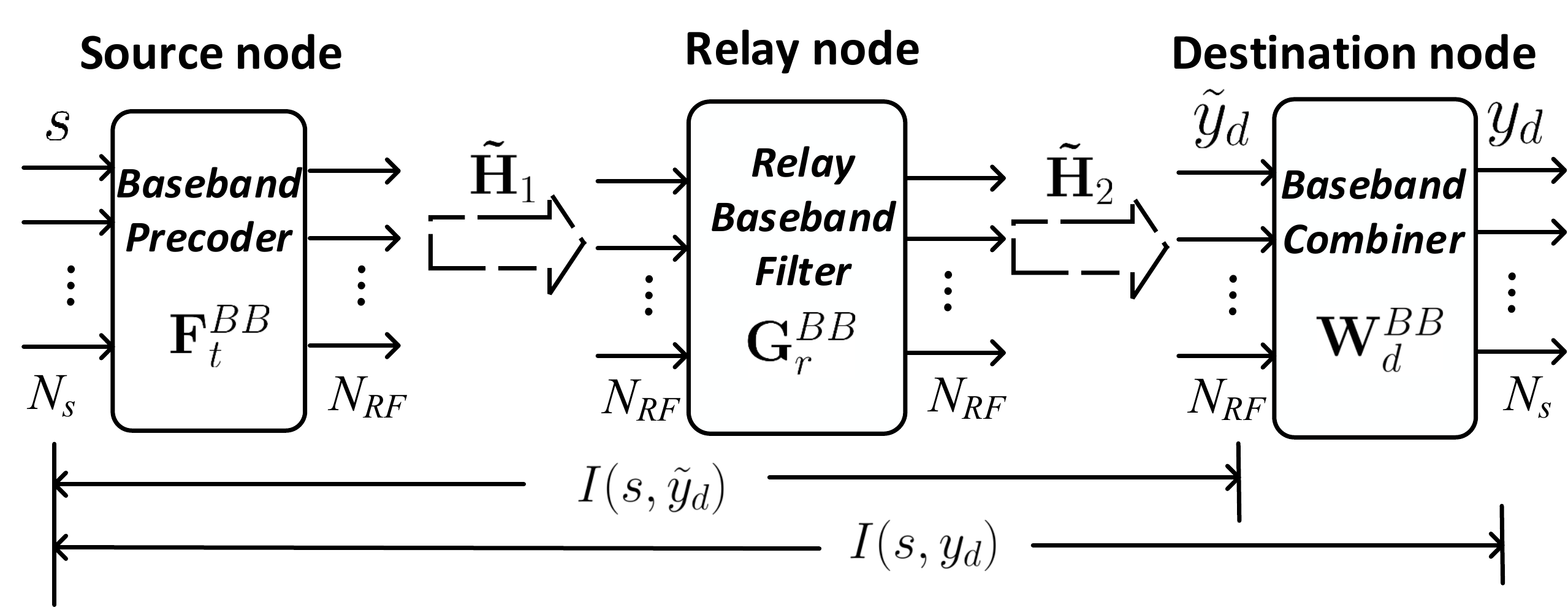}
\caption{Baseband System model}
\label{baseband system}
\end{figure*}

Using the equivalent channels (\ref{effechannel1}) and (\ref{effechannel2}), we rewrite the received signals at the destination node as
\begin{equation}
\tilde{\bm{y}}_\text{d} =\tilde{\mathbf{H}}_2\mathbf{G}_\text{r}^\text{BB}\tilde{\mathbf{H}}_1\mathbf{F}_\text{t}^\text{BB}\bm{s}+\tilde{\mathbf{H}}_2\mathbf{G}_\text{r}^\text{BB}\tilde{\bm{n}}_1+\tilde{\bm{n}}_2,
\end{equation}
\begin{equation}
\bm{y}_\text{d}=\mathbf{W}_\text{d}^\text{BB}\tilde{\mathbf{H}}_2\mathbf{G}_\text{r}^\text{BB}\tilde{\mathbf{H}}_1\mathbf{F}_\text{t}^\text{BB}\bm{s}+\mathbf{W}_\text{d}^\text{BB}\tilde{\mathbf{H}}_2\mathbf{G}_\text{r}^\text{BB}\tilde{\bm{n}}_1+\mathbf{W}_\text{d}^\text{BB}\tilde{\bm{n}}_2,
\end{equation}
where {$\tilde{\bm{n}}_1=\mathbf{W}_\text{r}^\text{RF}\bm{n}_1$} and {$\tilde{\bm{n}}_2=\mathbf{W}_\text{d}^\text{RF}\bm{n}_2$}. 

Our ultimate goal for the baseband design is to maximize the mutual information {$I(\bm{s},\bm{y}_\text{d})$}. However, directly optimizing {$I(\bm{s},\bm{y}_\text{d})$} is intractable. According to the data processing inequality \cite{gray2011entropy}, {$I(\bm{s},\bm{y}_\text{d}) \le I(\bm{s},\tilde{\bm{y}}_\text{d})$}. We first design {$\mathbf{F}_\text{t}^\text{BB}$ and $\mathbf{G}_\text{r}^\text{BB}$} to maximize the mutual information {$I(\bm{s},\tilde{\bm{y}}_\text{d})$}. After we get the maximum {$I(\bm{s},\tilde{\bm{y}}_\text{d})$}, we implement the MMSE-SIC for {$\mathbf{W}_\text{d}^\text{BB}$}, which according to \cite{tse2005fundamentals} is information lossless. In this way, we make {$I(\bm{s},\bm{y}_\text{d}) = I(\bm{s},\tilde{\bm{y}}_\text{d})$}. Since {$I(\bm{s},\tilde{\bm{y}}_\text{d})$} is maximized, {$I(\bm{s},\bm{y}_\text{d})$} is also maximized because of the data processing inequality and the independence of  {$I(\bm{s}; \tilde{\bm{y}}_\text{d})$} from {$\mathbf{W}_\text{d}^\text{BB}$}.

\subsection{$\mathbf{F}_\text{t}^\text{BB}$ and $\mathbf{G}_\text{r}^\text{BB}$ design}
In this section, we jointly design {$\mathbf{F}_\text{t}^\text{BB}$} and {$\mathbf{G}_\text{r}^\text{BB}$} to maximize {$I(\bm{s},\tilde{\bm{y}}_\text{d})$}. According to \cite{christensen2008weighted}, there exists a relationship between {$I(\bm{s},\tilde{\bm{y}}_\text{d})$} and the MSE matrix {$\mathbf{E}_\text{MMSE}$}, i.e., 
\begin{equation}
\label{relation}
I(\bm{s},\tilde{\bm{y}}_\text{d})= \log_2\det(\mathbf{E}_\text{MMSE}^{-1}),
\end{equation}
where the MSE matrix $\mathbf{E}_\text{MMSE}$ is defined as the mean square error covariance matrix given the MMSE receiver. The detailed proof can be found in \cite{christensen2008weighted}. We give a brief derivation procedure below.

The MMSE receiver is defined as 
\begin{equation}
\begin{split}
\label{mmsebb}
\mathbf{W}_\text{d}^\text{MMSE}&=\argmin \E[\|\mathbf{W}_\text{d}^\text{BB}\tilde{\bm{y}}_\text{d}-\bm{s}\|^2]\\
&=(\tilde{\mathbf{H}}_2\mathbf{G}_\text{r}^\text{BB}\tilde{\mathbf{H}}_1\mathbf{F}_\text{t}^\text{BB})^H(\tilde{\mathbf{H}}_2\mathbf{G}_\text{r}^\text{BB}\tilde{\mathbf{H}}_1\mathbf{F}_\text{t}^\text{BB}(\tilde{\mathbf{H}}_2\mathbf{G}_\text{r}^\text{BB}\tilde{\mathbf{H}}_1\mathbf{F}_\text{t}^\text{BB})^H+\mathbf{R}_{\tilde{\text{n}}})^{-1},
\end{split}
\end{equation}
where {$\mathbf{R}_{\tilde{\text{n}}}=\sigma_1^2\tilde{\mathbf{H}}_2\mathbf{G}_\text{r}^\text{BB}\mathbf{W}_\text{r}^\text{RF}(\tilde{\mathbf{H}}_2\mathbf{G}_\text{r}^\text{BB}\mathbf{W}_\text{r}^\text{RF})^H+\sigma_2^2\mathbf{W}_\text{r}^\text{RF}(\mathbf{W}_\text{r}^\text{RF})^H$}.

The MMSE matrix $\mathbf{E}_\text{MMSE}$ can be calculated by
\begin{equation}
\begin{split}
\label{EMMSE}
\mathbf{E}_\text{MMSE}&=\E[(\mathbf{W}_\text{d}^\text{MMSE}\tilde{\bm{y}}_\text{d}-\bm{s})(\mathbf{W}_\text{d}^\text{MMSE}\tilde{\bm{y}}_\text{d}-\bm{s})^H]\\&=(\mathbf{I}_{\text{N}_\text{s}} - \mathbf{W}_\text{d}^\text{MMSE}\tilde{\mathbf{H}}_2\mathbf{G}_\text{r}^\text{BB}\tilde{\mathbf{H}}_1\mathbf{F}_\text{t}^\text{BB})\\
&(\mathbf{I}_{\text{N}_\text{s}} - \mathbf{W}_\text{d}^\text{MMSE}\tilde{\mathbf{H}}_2\mathbf{G}_\text{r}^\text{BB}\tilde{\mathbf{H}}_1\mathbf{F}_\text{t}^\text{BB})^H + \mathbf{W}_\text{d}^\text{MMSE}\mathbf{R}_{\tilde{\text{n}}}(\mathbf{W}_\text{d}^\text{MMSE})^H.
\end{split}
\end{equation}
Substituting (\ref{mmsebb}) into (\ref{EMMSE}), we can express $\mathbf{E}_\text{MMSE}$ as
\begin{equation}
\label{EMMSE1}
\mathbf{E}_\text{MMSE}=(\mathbf{I}_{\text{N}_\text{s}}+(\tilde{\mathbf{H}}_2\mathbf{G}_\text{r}^\text{BB}\tilde{\mathbf{H}}_1\mathbf{F}_\text{t}^\text{BB})^H\mathbf{R}_{\tilde{\text{n}}}^{-1}\tilde{\mathbf{H}}_2\mathbf{G}_\text{r}^\text{BB}\tilde{\mathbf{H}}_1\mathbf{F}_\text{t}^\text{BB})^{-1}.
\end{equation}
From (\ref{EMMSE1}), we can obtain (\ref{relation}).

Based on (\ref{relation}), we can establish the equivalence between the {$I(\bm{s},\tilde{\bm{y}}_\text{d})$} maximization problem and a WMMSE problem as \cite{christensen2008weighted} did.

The {$I(\bm{s},\tilde{\bm{y}}_\text{d})$} maximization problem is formulated as
\begin{equation}
\label{P0}
\begin{split}
& \min_{\mathbf{F}_\text{t}^\text{BB},\mathbf{G}_\text{r}^\text{BB}} -I(\bm{s},\tilde{\bm{y}}_\text{d})\\
s.t.~&\left\|\mathbf{F}_\text{t}^\text{RF}\mathbf{F}_\text{t}^\text{BB}\right\|^2_F \leq \text{E}_\text{t},\\
&\E[\left\|\mathbf{F}_\text{r}^\text{RF}\mathbf{G}_\text{r}^\text{BB}\tilde{\mathbf{H}}_1\mathbf{F}_\text{t}^\text{BB}\bm{s}+\mathbf{F}_\text{r}^\text{RF}\mathbf{G}_\text{r}^\text{BB}\mathbf{W}_\text{r}^\text{RF}\bm{n}_1\right\|^2_F] \leq \text{E}_\text{r}.
\end{split}
\end{equation}

The WMMSE problem is formulated as
\begin{equation}
\label{P1}
\begin{split}
&\min_{\mathbf{F}_t^{BB},\mathbf{G}_r^{BB},\mathbf{V}}~\Tr(\mathbf{V} \mathbf{E}_{MMSE})\\
s.t.~&\left\|\mathbf{F}_t^{RF}\mathbf{F}_t^{BB}\right\|^2_F\le \text{E}_\text{t},\\
&\E[\left\|\mathbf{F}_r^{RF}\mathbf{G}_r^{BB}\tilde{\mathbf{H}}_1\mathbf{F}_t^{BB}\bm{s}+\mathbf{F}_r^{RF}\mathbf{G}_r^{BB}\mathbf{W}_r^{RF}\bm{n}_1\right\|^2_F]\le \text{E}_\text{r},
\end{split}
\end{equation}
where {$\mathbf{V}$} is a constant weight matrix. 

We will show that Problems (\ref{P0}) and (\ref{P1}) have the same optimum solution, i.e., the points that satisfy the KKT conditions for (\ref{P0}) and (\ref{P1}) are the same. Same as \cite{christensen2008weighted}, we set the partial derivatives of the Lagrange functions of (\ref{P0}) and (\ref{P1}) to zero. Note that the power constraints of (\ref{P0}) and (\ref{P1}) are the same. To prove the equivalence, we only need to calculate the partial derivatives of {$ -I(\bm{s},\tilde{\bm{y}}_\text{d})$} and {$tr(\mathbf{V} \mathbf{E}_\text{MMSE})$} w.r.t {$\mathbf{F}_\text{t}^\text{BB}$ and $\mathbf{G}_\text{r}^\text{BB}$}. Note that {$\partial\log\det(\mathbf{X}) = \Tr(\mathbf{X^{-1}}\partial\mathbf{X})$}. Taking {$\mathbf{F}_\text{t}^\text{BB}$} as an example, for {$ I(\bm{s},\tilde{\bm{y}}_\text{d})$}, we have
\begin{equation}
\frac{\partial -I(\bm{s},\tilde{\bm{y}}_\text{d}) }{\partial \mathbf{F}_\text{t}^\text{BB}} = -\frac{\partial\log_2\det(\mathbf{E}_\text{MMSE})}{\partial \mathbf{F}_\text{t}^\text{BB}}= - \frac{\Tr(\mathbf{{E}_\text{MMSE}^{-1}}\partial\mathbf{E}_\text{MMSE})}{\log_2\partial \mathbf{F}_\text{t}^\text{BB}},
\end{equation}

Note that {$\partial \mathbf{X}^{-1} = -\mathbf{X}^{-1}(\partial \mathbf{X})\mathbf{X}^{-1}$} and {$\partial (\mathbf{AX})= \partial (\mathbf{X})\mathbf{A} +\partial (\mathbf{A})\mathbf{X} $}. For {$tr(\mathbf{V} \mathbf{E}_\text{MMSE})$}, we have
\begin{equation}
\begin{split}
&\frac{\partial \Tr(\mathbf{V} \mathbf{E}_\text{MMSE}) }{\partial \mathbf{F}_\text{t}^\text{BB}} = -\frac{\Tr(\partial (\mathbf{V} (\mathbf{E}^{-1}_\text{MMSE})^{-1}))}{\partial \mathbf{F}_\text{t}^\text{BB}}\\
&= - \frac{\Tr(\mathbf{E}_\text{MMSE}\partial(\mathbf{E}^{-1}_\text{MMSE})\mathbf{E}_\text{MMSE}\mathbf{V}+\partial(\mathbf{V})\mathbf{E}_\text{MMSE})}{\partial \mathbf{F}_\text{t}^\text{BB}}\\
& =  - \frac{\Tr(\mathbf{E}_\text{MMSE}\partial(\mathbf{E}^{-1}_\text{MMSE})\mathbf{E}_\text{MMSE}\mathbf{V})}{\partial \mathbf{F}_\text{t}^\text{BB}}.
\end{split}
\end{equation}

If we set the constant weight matrix $\mathbf{V} = \frac{\mathbf{E}^{-1}_\text{MMSE}}{\log2}$, then we have
\begin{equation}
\label{eq1}
\frac{\partial -I(\bm{s},\tilde{\bm{y}}_\text{d}) }{\partial \mathbf{F}_\text{t}^\text{BB}} = \frac{\partial \Tr(\mathbf{V} \mathbf{E}_\text{MMSE}) }{\partial \mathbf{F}_\text{t}^\text{BB}}.
\end{equation}

Similarly, we can derive
\begin{equation}
\label{eq2}
\frac{\partial -I(\bm{s},\tilde{\bm{y}}_\text{d}) }{\partial \mathbf{G}_\text{r}^\text{BB}} = \frac{\partial \Tr(\mathbf{V} \mathbf{E}_\text{MMSE}) }{\partial \mathbf{G}_\text{r}^\text{BB}}.
\end{equation}

From Eqs. (\ref{eq1}) and (\ref{eq2}), we can conclude that  the KKT-conditions of (\ref{P0}) and (\ref{P1}) can be satisfied simultaneously, which suggests that it is possible to solve the mutual information maximization problem through the use of WMMSE by choosing an appropriate weight, i.e., $\mathbf{V}$. 
To maximize $I(\bm{s},\tilde{\bm{y}}_d)$, we propose an iterative algorithm based on the WMMSE problem (\ref{P1}). Note that in the proposed algorithm, we also iteratively search for the appropriate weight matrix. When the algorithm converges, we will obtain the desired weight matrix as well as the optimal filters that maximize $I(\bm{s},\tilde{\bm{y}}_\text{d})$. The detailed algorithm is described as follows:
%

\begin{enumerate} 
\item Calculate the MMSE receiver {$\mathbf{W}_\text{d}^\text{MMSE}$} in Eq. (\ref{mmsebb}) and the MSE matrix {$\mathbf{E}_\text{MMSE}$} in Eq. (\ref{EMMSE}).
\item Update {$\mathbf{V}$} by setting  {$\mathbf{V} = \frac{\mathbf{E}_\text{MMSE}^{-1}}{\log2}$}.
\item Fixing {$\mathbf{V}$} and {$\mathbf{F}_\text{t}^\text{BB}$}, then we find {$\mathbf{G}_\text{r}^\text{BB}$} that minimizes {$\Tr(\mathbf{VE}_\text{MMSE})=\\
\Tr(\mathbf{V}((\mathbf{I}_{\text{N}_\text{s}} - \mathbf{W}_\text{d}^\text{MMSE}\tilde{\mathbf{H}}_2\mathbf{G}_\text{r}^\text{BB}\tilde{\mathbf{H}}_1\mathbf{F}_\text{t}^\text{BB})(\mathbf{I}_{\text{N}_\text{s}} - \mathbf{W}_\text{d}^\text{MMSE}\tilde{\mathbf{H}}_2\mathbf{G}_\text{r}^\text{BB}\tilde{\mathbf{H}}_1\mathbf{F}_\text{t}^\text{BB})^H + \mathbf{W}_\text{d}^\text{MMSE}\mathbf{R}_{\tilde{\text{n}}}(\mathbf{W}_\text{d}^\text{MMSE})^H))$} under the power constraints, i.e.,
\begin{equation}
\label{subPGr}
\begin{split}
\hat{\mathbf{G}}_\text{r}^\text{BB} &=\argmin~\Tr(\mathbf{V}((\mathbf{I}_{\text{N}_\text{s}} - \mathbf{W}_\text{d}^\text{MMSE}\tilde{\mathbf{H}}_2\mathbf{G}_\text{r}^\text{BB}\tilde{\mathbf{H}}_1\mathbf{F}_\text{t}^\text{BB})\\
&(\mathbf{I}_{\text{N}_\text{s}} - \mathbf{W}_\text{d}^\text{MMSE}\tilde{\mathbf{H}}_2\mathbf{G}_\text{r}^\text{BB}\tilde{\mathbf{H}}_1\mathbf{F}_\text{t}^\text{BB})^H + \mathbf{W}_\text{d}^\text{MMSE}\mathbf{R}_{\tilde{\text{n}}}(\mathbf{W}_\text{d}^\text{MMSE})^H))\\
s.t.~&\E[\left\|\mathbf{F}_\text{r}^\text{RF}\mathbf{G}_\text{r}^\text{BB}\tilde{\mathbf{H}}_1\mathbf{F}_\text{t}^\text{BB}\bm{s}+\mathbf{F}_\text{r}^\text{RF}\mathbf{G}_\text{r}^\text{BB}\mathbf{W}_\text{r}^\text{RF}\bm{n}_1\right\|^2_F]\leq \text{E}_\text{r}.
\end{split}
\end{equation}
Problem (\ref{subPGr}) is a convex optimization for $\mathbf{G}_\text{r}^\text{BB}$ and we can solve it using the KKT condition. Denoting the Lagrange function of Problem (\ref{subPGr}) as $L^r(\mathbf{G}_\text{r}^\text{BB},\lambda^\text{r})=\Tr(\mathbf{VE}_\text{MMSE})+\lambda^\text{r}(\left\|\mathbf{F}_\text{r}^\text{RF}\mathbf{G}_\text{r}^\text{BB}\tilde{\mathbf{H}}_1\mathbf{F}_\text{t}^\text{BB}\bm{s}+\mathbf{F}_\text{r}^\text{RF}\mathbf{G}_\text{r}^\text{BB}\mathbf{W}_\text{r}^\text{RF}\bm{n}_1\right\|^2_F- E_\text{r})$, the KKT conditions are 
\begin{equation}
\label{lagrang}
\frac{\partial L^\text{r}(\mathbf{G}_\text{r}^\text{BB},\lambda^\text{r})}{\partial \mathbf{G}_\text{r}^\text{BB}} =0,
\end{equation}
\begin{equation}
\label{lagc1}
\E[\left\|\mathbf{F}_\text{r}^\text{RF}\mathbf{G}_\text{r}^\text{BB}\tilde{\mathbf{H}}_1\mathbf{F}_\text{t}^\text{BB}\bm{s}+\mathbf{G}_\text{r}^\text{BB}\mathbf{W}_\text{r}^\text{RF}\bm{n}_1\right\|^2_F]- \text{E}_\text{r}\leq 0,
\end{equation}
\begin{equation}
\label{lagc2}
\lambda^\text{r}(\E[\left\|\mathbf{F}_\text{r}^\text{RF}\mathbf{G}_\text{r}^\text{BB}\tilde{\mathbf{H}}_1\mathbf{F}_\text{t}^\text{BB}\bm{s}+\mathbf{G}_\text{r}^\text{BB}\mathbf{W}_\text{r}^\text{RF}\bm{n}_1\right\|^2_F]- \text{E}_\text{r})=0,
\end{equation}
\begin{equation}
\label{lagc3}
\lambda^\text{r}\geq 0.
\end{equation}
Solving (\ref{lagrang}), we have 
\begin{equation}
\label{GrBB}
\begin{split}
\hat{\mathbf{G}}_\text{r}^\text{BB} &= ((\mathbf{W}_\text{d}^\text{MMSE}\tilde{\mathbf{H}}_2)^H\mathbf{V}\mathbf{W}_\text{d}^\text{MMSE}\tilde{\mathbf{H}}_2+\lambda^\text{r}(\mathbf{F}_\text{r}^\text{RF})^H\mathbf{F}_\text{r}^\text{RF})^{-1}\\
&(\mathbf{\tilde{H}}_2)^H(\mathbf{W}_\text{d}^\text{MMSE})^H\mathbf{V}(\mathbf{F}_\text{t}^\text{BB})^H(\mathbf{\tilde{H}}_1)^H(\mathbf{\tilde{H}}_1\mathbf{F}_\text{t}^\text{BB}(\mathbf{\tilde{H}}_1\mathbf{F}_\text{t}^\text{BB})^H+\sigma_1^2\mathbf{W}_\text{r}^\text{RF}(\mathbf{W}_\text{r}^\text{RF})^H)^{-1}.
\end{split}
\end{equation}

Based on (\ref{lagc1}) and (\ref{lagc2}), we can obtain the Lagrange multiplier $\lambda^\text{r}$ as follows. First, we calculate $\hat{\mathbf{G}}_\text{r}^\text{BB}$ by setting $\lambda^\text{r}=0$. If the power constraint is satisfied, then we set $\lambda^\text{r}=0$. If the power constraint is not satisfied, then, we initialize $\lambda^\text{r}$ with a pre-defined value and substitute it into (\ref{lagc1}) and  start a bisection search for $\lambda^\text{r}$ until the power constraint is satisfied.  

\item Fixing {$\mathbf{V}$}  and  {$\mathbf{G}_\text{r}^\text{BB}$}, then we find the {$\mathbf{F}_\text{t}^\text{BB}$} to minimize {$\Tr(\mathbf{VE}_\text{MMSE})=\\
	\Tr(\mathbf{V}((\mathbf{I}_{\text{N}_\text{s}} - \mathbf{W}_\text{d}^\text{MMSE}\tilde{\mathbf{H}}_2\mathbf{G}_\text{r}^\text{BB}\tilde{\mathbf{H}}_1\mathbf{F}_\text{t}^\text{BB})(\mathbf{I}_{\text{N}_\text{s}} - \mathbf{W}_\text{d}^\text{MMSE}\tilde{\mathbf{H}}_2\mathbf{G}_\text{r}^\text{BB}\tilde{\mathbf{H}}_1\mathbf{F}_\text{t}^\text{BB})^H + \mathbf{W}_\text{d}^\text{MMSE}\mathbf{R}_{\tilde{\text{n}}}(\mathbf{W}_\text{d}^\text{MMSE})^H))$} under the power constraints, i.e.,
\begin{equation}
\label{subPFt}
\begin{split}
&\hat{\mathbf{F}}_\text{t}^\text{BB}=\argmin~\Tr(\mathbf{V}((\mathbf{I}_{\text{N}_\text{s}} - \mathbf{W}_\text{d}^\text{MMSE}\tilde{\mathbf{H}}_2\mathbf{G}_\text{r}^\text{BB}\tilde{\mathbf{H}}_1\mathbf{F}_\text{t}^\text{BB})\\
&(\mathbf{I}_{\text{N}_\text{s}} - \mathbf{W}_\text{d}^\text{MMSE}\tilde{\mathbf{H}}_2\mathbf{G}_\text{r}^\text{BB}\tilde{\mathbf{H}}_1\mathbf{F}_\text{t}^\text{BB})^H + \mathbf{W}_\text{d}^\text{MMSE}\mathbf{R}_{\tilde{\text{n}}}(\mathbf{W}_\text{d}^\text{MMSE})^H))\\
s.t.~&\left\|\mathbf{F}_\text{t}^\text{RF}\mathbf{F}_\text{t}^\text{BB}\right\|^2_F\leq \text{E}_\text{t},\\
&\E[\left\|\mathbf{F}_\text{r}^\text{RF}\mathbf{G}_\text{r}^\text{BB}\tilde{\mathbf{H}}_1\mathbf{F}_\text{t}^\text{BB}s+\mathbf{F}_\text{r}^\text{RF}\mathbf{G}_\text{r}^\text{BB}\mathbf{W}_\text{r}^\text{RF}\bm{n}_1\right\|^2_F]\leq \text{E}_\text{r}.
\end{split}
\end{equation}
	Problem (\ref{subPFt}) is a convex optimization for {$\mathbf{F}_\text{t}^\text{BB}$} and, similar with Problem (\ref{subPGr}), we can solve Problem (\ref{subPFt}) using the KKT conditions. Denoting the Lagrange function of Problem (\ref{subPFt}) as $L^\text{t}(\mathbf{F}_\text{t}^\text{BB},\lambda^\text{t}_1,\lambda^\text{t}_2)$, the KKT conditions are
\begin{equation}
\label{lagrangFt}
\frac{\partial L^\text{t}(\mathbf{F}_\text{t}^\text{BB},\lambda^\text{t}_1,\lambda^\text{t}_2)}{\partial \mathbf{F}_\text{t}^\text{BB}} =0,
\end{equation}
\begin{equation}
\label{lagFtc1}
\left\|\mathbf{F}_\text{t}^\text{RF}\mathbf{F}_\text{t}^\text{BB}\right\|^2_F-\text{E}_\text{t}\leq 0,
\end{equation}
\begin{equation}
\label{lagFtc2}
\E[\left\|\mathbf{F}_\text{r}^\text{RF}\mathbf{G}_\text{r}^\text{BB}\tilde{\mathbf{H}}_1\mathbf{F}_\text{t}^\text{BB}s+\mathbf{F}_\text{r}^\text{RF}\mathbf{G}_\text{r}^\text{BB}\mathbf{W}_\text{r}^\text{RF}\bm{n}_1\right\|^2_F]- \text{E}_\text{r}\leq 0,
\end{equation}
\begin{equation}
\label{lagFtc3}
\lambda^\text{t}_1(\left\|\mathbf{F}_\text{t}^\text{RF}\mathbf{F}_\text{t}^\text{BB}\right\|^2_F-\text{E}_\text{t})=0,
\end{equation}
\begin{equation}
\label{lagFtc4}
\lambda^\text{t}_2(\E[\left\|\mathbf{F}_\text{r}^\text{RF}\mathbf{G}_\text{r}^\text{BB}\tilde{\mathbf{H}}_1\mathbf{F}_\text{t}^\text{BB}s+\mathbf{F}_\text{r}^\text{RF}\mathbf{G}_\text{r}^\text{BB}\mathbf{W}_\text{r}^\text{RF}\bm{n}_1\right\|^2_F]- \text{E}_\text{r})=0,
\end{equation}
\begin{equation}
\label{lagc3}
\lambda^\text{t}_1,\lambda^\text{t}_2\geq0.
\end{equation}

The optimal solution for {$\mathbf{F}_\text{t}^\text{BB}$} can be expressed as
\begin{equation}
\label{FtBB}
\begin{split}
&\hat{\mathbf{F}}_\text{t}^\text{BB}=((\mathbf{W}_\text{d}^\text{MMSE}\tilde{\mathbf{H}}_2\mathbf{G}_\text{r}^\text{BB}\tilde{\mathbf{H}}_1)^H\mathbf{V}\mathbf{W}_\text{d}^\text{MMSE}\tilde{\mathbf{H}}_2\mathbf{G}_\text{r}^\text{BB}\tilde{\mathbf{H}}_1 \\
&+\lambda^\text{t}_1(\mathbf{F}_\text{t}^\text{RF})^H\mathbf{F}_\text{t}^\text{RF}+ \lambda^\text{t}_2(\mathbf{F}_\text{r}^\text{RF}\mathbf{G}_\text{r}^\text{BB}\tilde{\mathbf{H}}_1)^H\mathbf{F}_\text{r}^\text{RF}\mathbf{G}_\text{r}^\text{BB}\tilde{\mathbf{H}}_1)^{-1}(\mathbf{V}\mathbf{W}_\text{d}^\text{MMSE}\tilde{\mathbf{H}}_2\mathbf{G}_\text{r}^\text{BB}\tilde{\mathbf{H}}_1)^H,
\end{split} 
\end{equation}
where {$\lambda^\text{t}_1$} and {$\lambda^\text{t}_2$} are the non-negative Lagrange multipliers corresponding to the power constraints. Similar with Problem (\ref{subPGr}), to obtain {$\lambda^\text{t}_1$} and {$\lambda^\text{t}_2$}, we consider four cases: i) if the power constraints are satisfied when {$\lambda^\text{t}_1=0$} and {$\lambda^\text{t}_2=0$}, we will set {$\lambda^\text{t}_1$} and {$\lambda^\text{t}_2$} equal to $0$; ii) if case (i) is not satisfied, we then set $\lambda^\text{t}_1=0$ and start a bisection search for $\lambda^\text{t}_2$ until the KKT condition (\ref{lagFtc4}) and the power constraint (\ref{lagFtc1}) are satisfied; iii) if  (\ref{lagFtc4}) and  (\ref{lagFtc1}) cannot be satisfied simultaneously through the bisection search for $\lambda^\text{t}_2$, we then set $\lambda^\text{t}_2=0$ and start a bisection search for $\lambda^\text{t}_1$ until (\ref{lagFtc2}) and (\ref{lagFtc3}) are satisfied;
iv) if we cannot find the appropriate $\lambda^\text{t}_1$ to satisfy (\ref{lagFtc2}) and (\ref{lagFtc3}) at the same time, we can obtain {$\lambda^\text{t}_1$} and {$\lambda^\text{t}_2$} through a two-layer bisection search. The search algorithm is described in Algorithm \ref{alg1}.
\begin{algorithm}[h]
\small
\caption{Two-layer bisection search for $\lambda^\text{t}_1$ and $\lambda^\text{t}_2$ }
\begin{algorithmic}[1]
\STATE initialize $\lambda^\text{t}_\text{1,min}=\lambda^\text{t}_\text{2,min}=0$, $\lambda^\text{t}_\text{1,max}$,$\lambda^\text{t}_\text{2,max}$;\
\WHILE{$\lambda^\text{t}_\text{1,max}-\lambda^\text{t}_\text{1,min}>\epsilon_1$} 
\STATE setting $\lambda^\text{t}_1 = \frac{\lambda^\text{t}_\text{1,min}+\lambda^\text{t}_\text{1,max}}{2}$;\
\WHILE {$\lambda^\text{t}_\text{2,max}-\lambda^\text{t}_\text{2,min}>\epsilon_2$} 
\STATE setting $\lambda^\text{t}_2 = \frac{\lambda^\text{t}_\text{2,min}+\lambda^\text{t}_\text{2,max}}{2}$;\ 
\STATE calculate $\mathbf{F}_\text{t}^\text{BB}$ according to (\ref{FtBB});\
\IF{$\E[\left\|\mathbf{F}_\text{r}^\text{RF}\mathbf{G}_\text{r}^\text{BB}\tilde{\mathbf{H}}_1\mathbf{F}_\text{t}^\text{BB}s+\mathbf{F}_\text{r}^\text{RF}\mathbf{G}_\text{r}^\text{BB}\mathbf{W}_\text{r}^\text{RF}\bm{n}_1\right\|^2_F]\leq \text{E}_\text{r}$}
\STATE $\lambda^\text{t}_\text{2,max} = \lambda^\text{t}_2$;\
\ENDIF
\IF{$\E[\left\|\mathbf{F}_\text{r}^\text{RF}\mathbf{G}_\text{r}^\text{BB}\tilde{\mathbf{H}}_1\mathbf{F}_\text{t}^\text{BB}s+\mathbf{F}_\text{r}^\text{RF}\mathbf{G}_\text{r}^\text{BB}\mathbf{W}_\text{r}^\text{RF}\bm{n}_1\right\|^2_F]\geq \text{E}_\text{r}$}
\STATE $\lambda^\text{t}_\text{2,min} = \lambda^\text{t}_2$;\
\ENDIF
\ENDWHILE
\STATE  calculate $\mathbf{F}_\text{t}^\text{BB}$ according to (\ref{FtBB});\
\IF{$\left\|\mathbf{F}_\text{t}^\text{RF}\mathbf{F}_\text{t}^\text{BB}\right\|^2_F\leq \text{E}_\text{t}$}
\STATE $\lambda^\text{t}_\text{1,max} = \lambda^\text{t}_1$;\
\ENDIF
\IF{$\left\|\mathbf{F}_\text{t}^\text{RF}\mathbf{F}_\text{t}^\text{BB}\right\|^2_F\geq \text{E}_\text{t}$}
\STATE $\lambda^\text{t}_\text{1,min} = \lambda^\text{t}_1$;\
\ENDIF
\ENDWHILE
\end{algorithmic}
\label{alg1}
\end{algorithm}
\end{enumerate}
The entire design procedures for $\mathbf{F}^\text{BB}_\text{t}$ and $\mathbf{G}^\text{BB}_\text{r}$ are summarized in  Algorithm \ref{alg2}.
\begin{algorithm}[h]
\small
\caption{Design for $\mathbf{F}^\text{BB}_\text{t}$ and $\mathbf{G}^\text{BB}_\text{r}$ }
\begin{algorithmic}[1]
\STATE Initialize $\mathbf{F}^{\text{BB}(0)}_\text{t}$ and $\mathbf{G}^{\text{BB}(0)}_{\text{r}}$;
\STATE Set $k=1$; 
\WHILE{$|I(\bm{s},\bm{\tilde{y}}_\text{d})^{(k)}-I(\bm{s},\bm{\tilde{y}}_\text{d})^{(k-1)}|>\epsilon$} 
\STATE Calculate the MMSE receiver $\mathbf{W}_\text{d}^{\text{MMSE}(k)}$ according to (\ref{mmsebb}) and the MSE matrix {$\mathbf{E}_\text{MMSE}^{(k)}$} in Eq. (\ref{EMMSE});
\STATE Update {$\mathbf{V}$} by setting  {$\mathbf{V} = \frac{(\mathbf{E}_\text{MMSE}^{(k)})^{-1}}{\log2}$};
\STATE Calculate $\mathbf{F}^{\text{BB}(k)}_\text{t}$ as Step III illustrates;
\STATE Calculate $\mathbf{G}^{\text{BB}(k)}_\text{r}$ as Step IV illustrates;
\STATE $k=k+1$;
\ENDWHILE
\end{algorithmic}
\label{alg2}
\end{algorithm}

\subsection{Convergence analysis}
Since the constant weight matrix {$\mathbf{V}$} changes at each iteration, it does not generate a monotonic decreasing sequence, which means we cannot directly prove the convergence of the proposed algorithm. Fortunately,  according to \cite{christensen2008weighted}, the iteration procedure to maximize the mutual information through minimizing WMMSE is the same optimization procedure for an equivalent optimization problem as below
\begin{equation}
\label{equivProb}
\begin{split}
\min_{\substack{\mathbf{F}_\text{t}^\text{BB},\mathbf{G}_\text{r}^\text{BB},\\\mathbf{V},\mathbf{W}_\text{d}^\text{BB}}} &\Tr(\mathbf{V}\mathbf{E}_\text{MMSE}) - \log\det(\mathbf{V})\\
s.t. ~&\left\|\mathbf{F}_\text{t}^\text{RF}\mathbf{F}_\text{t}^\text{BB}\right\|^2_F \leq \text{E}_\text{t},\\
&\E[\left\|\mathbf{F}_\text{r}^\text{RF}\mathbf{G}_\text{r}^\text{BB}\tilde{\mathbf{H}}_1\mathbf{F}_\text{t}^\text{BB}\bm{s}+\mathbf{F}_\text{r}^\text{RF}\mathbf{G}_\text{r}^\text{BB}\mathbf{W}_\text{r}^\text{RF}\bm{n}_1\right\|^2_F] \leq \text{E}_\text{r}.
\end{split}
\end{equation}
The proof of the equivalence is similar to the proof in \cite{christensen2008weighted} and we omit the detailed prove for brevity. The main idea is that the alternating minimization of the objective in (\ref{equivProb}) corresponds to Steps 1-4 in our proposed algorithm. For example, when $\mathbf{F}_\text{t}^\text{BB}$,$\mathbf{G}_\text{r}^\text{BB}$ and $\mathbf{V}$ are fixed, optimizing (\ref{equivProb}) w.r.t. $\mathbf{W}_\text{d}^\text{BB}$ becomes minimizing MMSE, which gives the same result as Step 1. When $\mathbf{F}_\text{t}^\text{BB}$,$\mathbf{G}_\text{r}^\text{BB}$ and $\mathbf{W}_\text{d}^\text{BB}$ are fixed, the optimal solution for $\mathbf{V}$ which minimizes the objective $\Tr(\mathbf{V}\mathbf{E}_\text{MMSE}) - \log\det(\mathbf{V})$ in (\ref{equivProb}) is the same as Step 2. 

Based on this equivalence, we can prove the convergence of the proposed algorithm by proving the monotonic convergence of Problem (\ref{equivProb}). According to \cite{christensen2008weighted}, the objective in Problem (\ref{equivProb}) has a lower bound, which is the negative of the maximum mutual information. Due to the alternating minimization process, the objective in Problem (\ref{equivProb}) decreases monotonically. Since a sequence of monotonically decreasing numbers with a lower bound converges, Problem (\ref{equivProb}) converges and so does our proposed algorithm.

\subsection{Complexity analysis}
Since we provide closed-form solutions for each iteration, the main complexity lies in the search for the appropriate Lagrange multipliers. Let us define the search accuracy as $\epsilon$. This is a relative measure for the search interval. For example, if the length of our search interval is $l$, then the threshold for the search termination is set to be $\epsilon l$. Based on the accuracy $\epsilon$, the number of iterations in the bisection search in Step 3 is bounded by $\mathcal{O}(\log_2\frac{1}{\epsilon})$. In Step 4, we use a two-layer bisection search, whose number of iteration is  bounded by $\mathcal{O}(\log_2^2\frac{1}{\epsilon})$. So, for each outer iteration, the total number of inner iterations is  $\mathcal{O}(\log_2\frac{1}{\epsilon})+\mathcal{O}(\log_2^2\frac{1}{\epsilon})$. Compared with the algorithm in \cite{xue2018relay}, for which the number of inner iterations is $\mathcal{O}(2(2\text{N}_\text{RF}\text{N}_\text{r})^{2.5}\log\frac{1}{\epsilon})$ for each outer iteration, the complexity of our algorithm is much lower especially for large antenna arrays.

\subsection{$\mathbf{W}_\text{d}^\text{BB}$ design}
Since {$I(\bm{s},\bm{y}_\text{d}) \le I(\bm{s},\tilde{\bm{y}}_\text{d})$} \cite{gray2011entropy}, after we find the maximum {$I(\bm{s},\tilde{\bm{y}}_\text{d})$}, the optimal {$I(\bm{s},\bm{y}_\text{d})$} will be obtained if the destination node baseband processing does not cause any information loss. According to \cite{tse2005fundamentals}, MMSE-SIC is information lossless. Therefore, we use MMSE-SIC for our  destination baseband design. To simplify the expression, let us define 
\begin{equation}
\tilde{\bm{y}}_\text{d} =\tilde{\mathbf{H}}_2\mathbf{G}_\text{r}^\text{BB}\tilde{\mathbf{H}}_1\mathbf{F}_\text{t}^\text{BB}\bm{s}+\tilde{\mathbf{H}}_2\mathbf{G}_\text{r}^\text{BB}\tilde{\bm{n}}_1+\tilde{\bm{n}}_2 = \mathbf{G}\bm{s}+\bar{\bm{v}},
\end{equation}
where {$\mathbf{G}=[\bm{g}_1,...,\bm{g}_{\text{N}_\text{s}}] \in \mathbf{C}^{\text{N}_\text{RF}\times \text{N}_\text{s}}$, $\bar{\bm{v}}$} is the colored noise with covariance matrix {$\mathbf{R}_{\tilde{\text{n}}}$}.

To implement the MMSE-SIC for the $k^\text{th}$ stream, we subtract the effect of the first $k-1$ streams from the output and obtain
\begin{equation}
\tilde{\bm{y}}_\text{d'} = \tilde{\bm{y}}_\text{d} - \sum_{i=1}^{k-1}\bm{g}_is_i + \bar{\bm{v}} = \bm{g}_ks_k + \sum_{j=k+1}^{\text{N}_\text{s}}\bm{g}_js_j + \bar{\bm{v}}.
\end{equation}
Define {$\mathbf{W}_\text{d}^\text{BB}= [\bm{w}_1,...,\bm{w}_{\text{N}_\text{s}}]^H$}, the baseband filter for the $k^{th}$ stream is derived as the MMSE filter:
\begin{equation}
\bm{w}_k^H = \bm{g}_k^H(\sum_{j=k+1}^{\text{N}_\text{s}}\bm{g}_j\bm{g}_j^H+\mathbf{R}_{\tilde{\text{n}}})^{-1}.
\end{equation}

\section{Robust Design}
So far, we have designed the mmWave relay precoders/combiners under the perfect channel information. However, it is hard to avoid estimation/quantization errors while obtaining the channel information. To study the effects of imperfect channel estimation, we adopt the model provided in \cite{zhang2008statistically,rong2011robust, xing2010robust}. In this model, the relationship between the channel values and the corresponding estimated channel values are:
\begin{equation}
\label{imchan1}
\mathbf{H}_1 = \mathbf{\bar{H}}_1 + \Phi_1^{\frac{1}{2}}\Delta_1\Theta_1^{\frac{1}{2}},
\end{equation}
\begin{equation}
\label{imchan2}
\mathbf{H}_2 = \mathbf{\bar{H}}_2 + \Phi_2^{\frac{1}{2}}\Delta_2\Theta_2^{\frac{1}{2}},
\end{equation}
where $\mathbf{H}_1$ and $\mathbf{H}_2$ are the actual channel matrices, i.e., the channels that we cannot precisely estimate,  and $\mathbf{\bar{H}}_1$ and $\mathbf{\bar{H}}_2$ are the estimated channels. The transmitting covariance matrix of the channel estimation error at the source node and the relay node are denoted by $\Theta_1 \in \mathbb{C}^{\text{N}_\text{t}\times \text{N}_\text{t}}$ and $\Theta_2\in \mathbb{C}^{\text{N}_\text{r}\times \text{N}_\text{r}}$, respectively. The receiving covariance matrix of the channel estimation error at the relay node and the destination node are denoted by $\Phi_1\in \mathbb{C}^{\text{N}_\text{r}\times \text{N}_\text{r}}$ and $\Phi_2 \in\mathbb{C}^{\text{N}_\text{d}\times \text{N}_\text{d}}$, respectively. $\Delta_1$ and $\Delta_2$ are Gaussian random matrices with independent and identically distributed (i.i.d.) zero mean and unit variance entries and are the unknown parts of the CSI mismatch.  we adopt the exponential model\cite{xing2010robust,rong2011robust} for the channel estimation error covariance matrices $\Phi_1$, $\Theta_1$, $\Phi_2$ and $\Theta_2$. To be specific, the entries of the matrices are given as $\Phi_1(i,j)=\sigma_{e,1}^2\beta_1^{|i-j|}$, $\Theta_1(i,j) = \alpha_1^{|i-j|}$, $\Phi_2(i,j)=\sigma_{e,2}^2\beta_2^{|i-j|}$ and $\Theta_2(i,j) = \alpha_2^{|i-j|}$, where $\alpha_1$, $\beta_1$, $\alpha_2$ and $\beta_2$ are the correlation coefficients and $\sigma_{e,1}^2$ and $\sigma_{e,2}^2$ denote the estimation error covariance. For simplicity, we assume $\alpha_1=\alpha_2=\alpha$, $\beta_1=\beta_2=\beta$ and $\sigma_{e,1}^2=\sigma_{e,2}^2=\sigma_e^2$.

As shown in Section \uppercase\expandafter{\romannumeral5}, the imperfect channel information will result in severe performance degradation. For example, the achievable data rate of \cite{single_relay_mmwave,xue2018relay} can be decreased to half of what it is for the perfect CSI.

To further increase the robustness of our proposed algorithm, in this section, we will propose a robust precoding/combining design for the mmWave relay system based on our proposed WMMSE algorithm. 

\subsection{RF design}
Recall that our RF precoding/combining is actually phase compensation for each eigenmode. The eigenmodes are obtained through SVD. When considering the imperfect CSI, the phase of each eigenmode cannot be precisely acquired. Let us take the actual channel $\mathbf{H}_1$ and the estimated $\bar{\mathbf{H}}_1$ as an example. The left singular matrices of $\mathbf{H}_1$ and $\bar{\mathbf{H}}_1$ are denoted by $\mathbf{U}_1$ and $\bar{\mathbf{U}}_1$. We denote the phase of each entry in $\mathbf{U}_1$ and $\bar{\mathbf{U}}_1$ by $\theta_{i,j}$ and $\bar{\theta}_{i,j}$, respectively. The phase difference in each entry can be calculated as $\Delta\theta_{i,j} =\theta_{i,j}-\bar{\theta}_{i,j}$. Let us assume that $\Delta\theta_{i,j}$ has a distribution $p(\Delta\theta)$. We want to make an estimation on $\Delta\theta_{i,j}$ to minimize the mean square estimation error $\E[(\Delta\hat{\theta}_{i,j}-\Delta\theta_{i,j})^2]$. The estimation $\Delta\hat{\theta}_{i,j}=\E[\Delta\theta_{i,j}]$ can be calculated based on the distribution $p(\Delta\theta)$. Note that we can only obtain the estimated channel. Once we calculate $\Delta\hat{\theta}_{i,j}$, we can calibrate the phase of each entry in $\bar{\mathbf{U}}_1$ as $\hat{\theta}_{i,j}=\bar{\theta}_{i,j}+\Delta\hat{\theta}_{i,j}$. Using the same approach, we can calibrate the phase of singular matrices of $\bar{\mathbf{H}}_1$ and $\bar{\mathbf{H}}_2$. Then, based on (\ref{c1}) and (\ref{c2}), we can calculate the RF precoders and combiners $\mathbf{\bar{F}}_t^\text{RF}$, $\mathbf{\bar{W}}_r^\text{RF}$, $\mathbf{\bar{F}}_\text{r}^\text{RF}$ and $\mathbf{\bar{W}}_\text{d}^\text{RF}$ based on the calibrated singular matrices of $\bar{\mathbf{H}}_1$ and $\bar{\mathbf{H}}_2$.

As we analyzed above, to calculate the RF precoders and combiners, we must  know the distribution of $\Delta\theta_{i,j}$ to make the estimation $\Delta\hat{\theta}_{i,j}=\E[\Delta\theta_{i,j}]$. However, the theoretical analysis for the phase distribution is intractable. To obtain the phase distribution, we simulate 100 channel realizations based on the imperfect channel models in (\ref{imchan1}), where we set  $\text{N}_\text{r}=32$, $\text{N}_\text{t}=48$ and $L=20$. We adopt the correlation model from \cite{zhang2008statistically,rong2011robust, xing2010robust} where the entries of the correlation matrices are selected as $\Phi_1(i,j)=\sigma_{e,1}^2\beta_1^{|i-j|}$, $\Theta_1(i,j) = \alpha_1^{|i-j|}$. In the simulation, we set $\alpha_1=0$, $\beta_1=0$ and $\sigma_{e,1}^2=0.1$.

We collect the phase difference in each matrix entry from 60 simulations. In Fig. \ref{robust_theta}, we plot the simulated probability density function (PDF) of the $\Delta \theta$ in solid line. We use a generalized normal distribution \cite{nadarajah2005generalized} to approximate the distribution. The PDF of a generalized normal distribution is expressed as $f(x)=\frac{\beta}{2\alpha\Gamma(\frac{1}{\beta})}e^{-(|x-\mu|/\alpha)^\beta}$. We can see the approaching effect of different value of the shaping parameter $\beta$ in Fig. \ref{robust_theta}. We use Kullback-Leibler distance as a performance measure for the approximation, which is calculated by $D_\text{KL}(\bm{Y}||\bm{X})=\sum_{i=1}^N\log(\frac{Y_i}{X_i})Y_i$ where $\bm{Y}$ and $\bm{X}$ are the probability distributions. The lower the Kullback-Leibler distance, the closer the two distributions are. Note that the absolute value of KL-distance varies when the number of total points (i.e., $N$) changes. In our simulation, the best approximation comes with the one with $\beta=2$, since it has the lowest Kullback-Leibler distance. When $\beta=2$, the generalized normal distribution in Fig. \ref{robust_theta} is a Gaussian distribution with 0 mean, which means we can estimate $\Delta\hat{\theta}_{i,j}=0$. 

\begin{figure}[htbp]
\centering
\includegraphics[width=10cm]{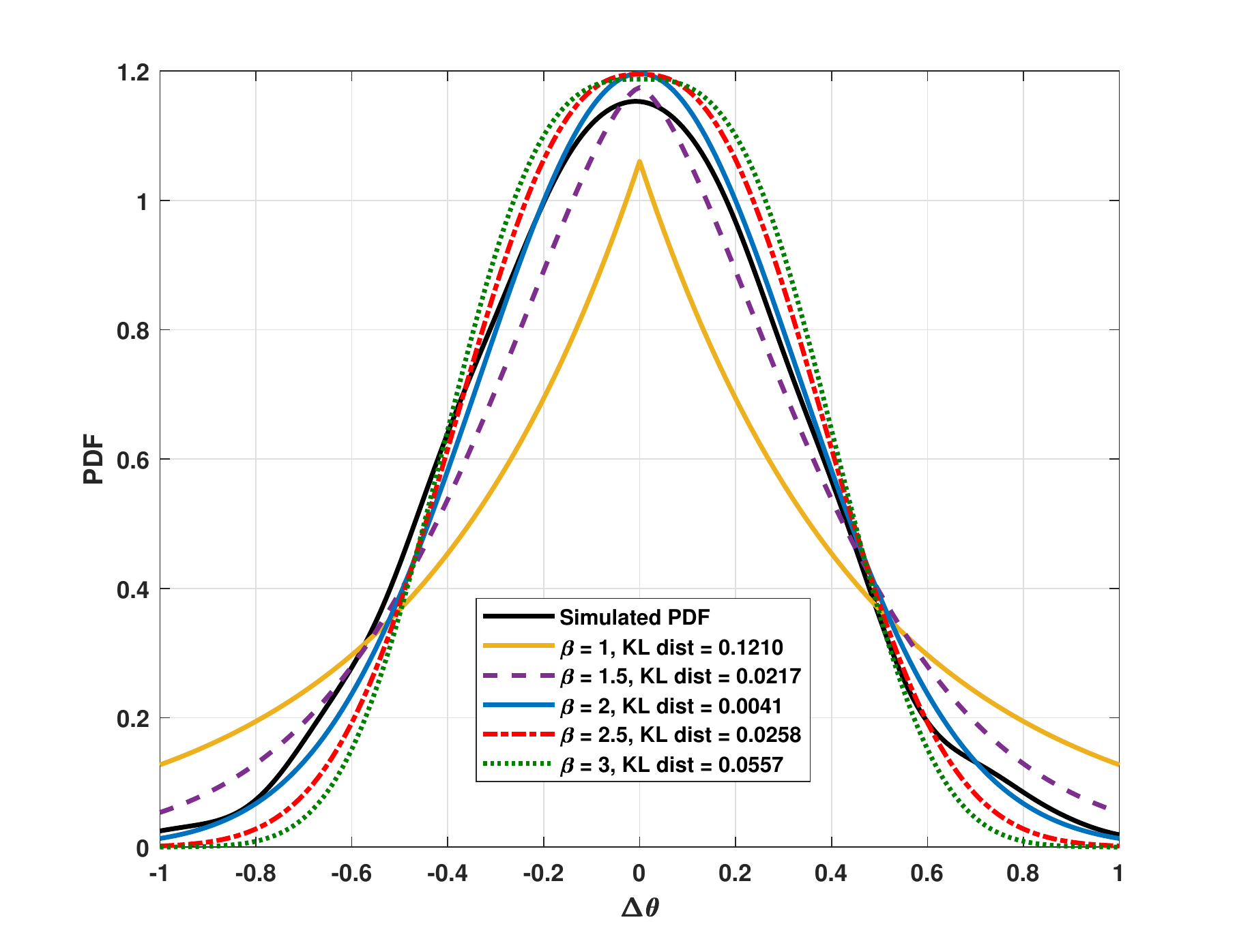}
\caption{ Approximation of the Simulated PDF}
\label{robust_theta}
\vspace*{-1cm}
\end{figure}

\subsection{Baseband design} 

Based on the RF precoders and combiners $\mathbf{\bar{F}}_\text{t}^\text{RF}$, $\mathbf{\bar{W}}_\text{r}^\text{RF}$, $\mathbf{\bar{F}}_\text{r}^\text{RF}$ and $\mathbf{\bar{W}}_\text{d}^\text{RF}$ designed in the last subsection, the equivalent baseband channels after the RF processing are
\begin{equation}
\label{robchannel1}
\begin{split}
\mathbf{\tilde{H}}_1 &= \mathbf{\bar{W}}_\text{r}^\text{RF}\mathbf{H}_1\mathbf{\bar{F}}_\text{t}^\text{RF}=\mathbf{\bar{W}}_\text{r}^\text{RF}\mathbf{\bar{H}}_1\mathbf{\bar{F}}_\text{t}^\text{RF} + \mathbf{\bar{W}}_\text{r}^\text{RF}\Phi_1^{\frac{1}{2}}\Delta_1\Theta_1^{\frac{1}{2}}\mathbf{\bar{F}}_\text{t}^\text{RF}\\
&=\mathbf{\tilde{\bar{H}}}_1 + \tilde{\Phi}_1^{\frac{1}{2}}\Delta_1\tilde{\Theta}_1^{\frac{1}{2}},
\end{split}
\end{equation}
\begin{equation}
\label{robchannel2}
\begin{split}
\mathbf{\tilde{H}}_2 &= \mathbf{\bar{W}}_\text{d}^\text{RF}\mathbf{H}_2\mathbf{\bar{F}}_\text{r}^\text{RF}=\mathbf{\bar{W}}_\text{d}^\text{RF}\mathbf{\bar{H}}_2\mathbf{\bar{F}}_\text{r}^\text{RF} + \mathbf{\bar{W}}_\text{d}^\text{RF}\Phi_2^{\frac{1}{2}}\Delta_2\Theta_2^{\frac{1}{2}}\mathbf{\bar{F}}_\text{r}^\text{RF}\\
&=\mathbf{\tilde{\bar{H}}}_2 + \tilde{\Phi}_2^{\frac{1}{2}}\Delta_2\tilde{\Theta}_2^{\frac{1}{2}},
\end{split}
\end{equation}
where we denote the true equivalent baseband channels by $\mathbf{\tilde{H}}_1$ and $\mathbf{\tilde{H}}_2$. The estimated equivalent baseband channels are denoted by $\mathbf{\tilde{\bar{H}}}_1$ and $\mathbf{\tilde{\bar{H}}}_2$, which are the channels we obtain at the source node. We define $\tilde{\Phi}_1:=\mathbf{\bar{W}}_\text{r}^\text{RF}\Phi_1(\mathbf{\bar{W}}_\text{r}^\text{RF})^H$, $\tilde{\Phi}_2:=\mathbf{\bar{W}}_\text{d}^\text{RF}\Phi_2(\mathbf{\bar{W}}_\text{d}^\text{RF})^H$, $\tilde{\Theta}_1:=(\mathbf{\bar{F}}_\text{t}^\text{RF})^H\Theta_1\mathbf{\bar{F}}_\text{t}^\text{RF}$ and $\tilde{\Theta}_2:=(\mathbf{\bar{F}}_\text{r}^\text{RF})^H\Theta_2\mathbf{\bar{F}}_\text{r}^\text{RF}$. 

For (\ref{robchannel1}) and (\ref{robchannel2}) , we have the following properties \cite{gupta2018matrix} (Using $\mathbf{\tilde{H}}_1$ as an example)
\begin{equation}
\label{property1}
\E_{\Delta_1}[\mathbf{\tilde{H}}_1\mathbf{C}\mathbf{\tilde{H}}_1^H] = \mathbf{\tilde{\bar{H}}}_1\mathbf{C} \mathbf{\tilde{\bar{H}}}_1^H + \Tr(\mathbf{C}\tilde{\Theta}_1)\tilde{\Phi}_1,
\end{equation}

\begin{equation}
\label{property2}
\E_{\Delta_1}[\mathbf{\tilde{H}}_1^H\mathbf{C}\mathbf{\tilde{H}}_1] = \mathbf{\tilde{\bar{H}}}_1^H\mathbf{C} \mathbf{\tilde{\bar{H}}}_1 + \Tr(\tilde{\Phi}_1\mathbf{C})\tilde{\Theta}_1.
\end{equation}

To design a robust baseband system, we need to redesign the algorithm in Section III based on the imperfect channel models (\ref{robchannel1}) and (\ref{robchannel2}). The main idea is similar, i.e., that we first optimize the baseband filters $\mathbf{\bar{F}}_\text{t}^\text{BB}$ and $\mathbf{\bar{G}}_\text{r}^\text{BB}$ to maximize the average $\E_{\Delta_1,\Delta_2}[I(\bm{s},\tilde{\bm{y}}_\text{d})]$ , and then we use MMSE-SIC for $\mathbf{\bar{W}}_\text{d}^\text{BB}$. Note that we denote the baseband precoder/combiner based on the estimated equivalent baseband channels $\mathbf{\tilde{\bar{H}}}_1$ and $\mathbf{\tilde{\bar{H}}}_2$ by $\mathbf{\bar{F}}_\text{t}^\text{BB}$, $\mathbf{\bar{G}}_\text{r}^\text{BB}$ and $\mathbf{\bar{W}}_\text{d}^\text{BB}$.
The main challenge here is if there still exists an equivalent relationship between the average mutual information maximization and the WMMSE minimization.

To derive the equivalent relationship, we first derive an upper bound for the average mutual information $\E_{\Delta_1,\Delta_2}[I(\bm{s},\tilde{\bm{y}}_\text{d})]$ as
\begin{equation}
\label{avgmutual}
\begin{split}
&\E^\text{UB}_{\Delta_1,\Delta_2}[I(\bm{s},\tilde{\bm{y}}_\text{d})] = \log_2\det(\E_{\Delta_1,\Delta_2}[\mathbf{I}_{\text{N}_\text{s}}+(\tilde{\mathbf{H}}_2\mathbf{\bar{G}}_\text{r}^\text{BB}\tilde{\mathbf{H}}_1\mathbf{\bar{F}}_\text{t}^\text{BB})^H\mathbf{R}_{\tilde{\text{n}}}^{-1}\tilde{\mathbf{H}}_2\mathbf{\bar{G}}_\text{r}^\text{BB}\tilde{\mathbf{H}}_1\mathbf{\bar{F}}_\text{t}^\text{BB}])\\
&=\log_2\det(\mathbf{I}_{\text{N}_\text{s}}+(\tilde{\bar{\mathbf{H}}}_2\mathbf{\bar{G}}_\text{r}^\text{BB}\tilde{\bar{\mathbf{H}}}_1\mathbf{\bar{F}}_\text{t}^\text{BB})^H\mathbf{R}_{\tilde{\bar{\text{n}}}}^{-1}\tilde{\bar{\mathbf{H}}}_2\mathbf{\bar{G}}_\text{r}^\text{BB}\tilde{\bar{\mathbf{H}}}_1\mathbf{\bar{F}}_\text{t}^\text{BB}+\mathbf{B}_1+\mathbf{B}_2),
\end{split}
\end{equation}
where 
\[
\mathbf{R}_{\tilde{\bar{\text{n}}}}:=\sigma_2^2\mathbf{\bar{W}}_\text{r}^\text{RF}(\mathbf{\bar{W}}_\text{r}^\text{RF})^H +\sigma_1^2(\mathbf{\tilde{\bar{H}}}_2\mathbf{\bar{G}}_\text{r}^\text{BB}\mathbf{\bar{W}}_\text{r}^\text{RF}(\mathbf{\tilde{\bar{H}}}_2\mathbf{\bar{G}}_\text{r}^\text{BB}\mathbf{\bar{W}}_\text{r}^\text{RF})^H+\Tr(\mathbf{\bar{G}}_\text{r}^\text{BB}\mathbf{\bar{W}}_\text{r}^\text{RF}(\mathbf{\bar{G}}_\text{r}^\text{BB}\mathbf{\bar{W}}_\text{r}^\text{RF})^H)\tilde{\Theta}_2)\tilde{\Phi}_2),
\]
\[
\mathbf{B}_1:=(\mathbf{\bar{G}}_\text{r}^\text{BB}\tilde{\mathbf{H}}_1\mathbf{\bar{F}}_\text{t}^\text{BB})^H\Tr(\tilde{\Phi}_2\mathbf{R}_{\tilde{\bar{\text{n}}}})\tilde{\Theta}_2\mathbf{\bar{G}}_\text{r}^\text{BB}\tilde{\mathbf{H}}_1\mathbf{\bar{F}}_\text{t}^\text{BB},
\]
\[
\mathbf{B}_2:=(\mathbf{\bar{F}}_\text{t}^\text{BB})^H\Tr(\tilde{\Phi}_1((\tilde{\bar{\mathbf{H}}}_2\mathbf{\bar{G}}_\text{r}^\text{BB})^H\mathbf{R}_{\tilde{\bar{\text{n}}}}\tilde{\bar{\mathbf{H}}}_2\mathbf{\bar{G}}_\text{r}^\text{BB}+(\mathbf{\bar{G}}_\text{r}^\text{BB})^H\Tr(\tilde{\Phi}_2\mathbf{R}_{\tilde{\bar{\text{n}}}})\tilde{\Theta}_2)\mathbf{\bar{G}}_\text{r}^\text{BB})\tilde{\Theta}_1\mathbf{\bar{F}}_\text{t}^\text{BB}.
\]

Then, we need to derive the expression of $\E_{\Delta_1,\Delta_2}[\mathbf{E}_\text{MMSE}] $. According to (\ref{EMMSE1}), we have
\begin{equation}
\label{robEMMSE}
\begin{split}
&\E_{\Delta_1,\Delta_2}[\mathbf{E}_\text{MMSE}] = \E_{\Delta_1,\Delta_2}[(\mathbf{I}_{\text{N}_\text{s}}+(\tilde{\mathbf{H}}_2\mathbf{\bar{G}}_\text{r}^\text{BB}\tilde{\mathbf{H}}_1\mathbf{\bar{F}}_\text{t}^\text{BB})^H\mathbf{R}_{\tilde{\text{n}}}^{-1}\tilde{\mathbf{H}}_2\mathbf{\bar{G}}_\text{r}^\text{BB}\tilde{\mathbf{H}}_1\mathbf{\bar{F}}_\text{t}^\text{BB})^{-1}]\\
&=(\mathbf{I}_{\text{N}_\text{s}}+(\tilde{\bar{\mathbf{H}}}_2\mathbf{\bar{G}}_\text{r}^\text{BB}\tilde{\bar{\mathbf{H}}}_1\mathbf{\bar{F}}_\text{t}^\text{BB})^H\mathbf{R}_{\tilde{\bar{\text{n}}}}^{-1}\tilde{\bar{\mathbf{H}}}_2\mathbf{\bar{G}}_\text{r}^\text{BB}\tilde{\bar{\mathbf{H}}}_1\mathbf{\bar{F}}_\text{t}^\text{BB}+\mathbf{B}_1+\mathbf{B}_2)^{-1}.
\end{split}
\end{equation}
Equation (\ref{robEMMSE}) implies that the relationship $\E^\text{UB}_{\Delta_1,\Delta_2}[I(\bm{s},\tilde{\bm{y}}_d)] = \log_2\det(\E_{\Delta_1,\Delta_2}[\mathbf{E}_\text{MMSE}]^{-1})$ still holds for the imperfect channel model, which means we can maximize the upper bound of the average mutual information through the WMMSE minimization as discussed in Section III.

The expression of $\E_{\Delta_1,\Delta_2}[\mathbf{E}_\text{MMSE}] $ in (\ref{robEMMSE}) includes a matrix inverse operator, which complicates the following calculation for $\mathbf{\bar{F}}_\text{t}^\text{BB}$ and $\mathbf{\bar{G}}_\text{r}^\text{BB}$. To derive a simpler expression for $\E_{\Delta_1,\Delta_2}[\mathbf{E}_\text{MMSE}]$, we first calculate the average MSE matrix $\E_{\Delta_1,\Delta_2}[\mathbf{E}_\text{MSE}]$. The MSE matrix  is given by
\begin{equation}
\label{EMSE}
\begin{split}
&\E_{\Delta_1,\Delta_2}[\mathbf{E}_\text{MSE}] = \\
&\mathbf{\bar{W}}_\text{d}^\text{BB}(\mathbf{A}+\mathbf{R}_{\tilde{\bar{\text{n}}}}) (\mathbf{\bar{W}}_\text{d}^\text{BB})^H - (\mathbf{\bar{W}}_\text{d}^\text{BB})^H\mathbf{\tilde{\bar{H}}}_2\mathbf{\bar{G}}_\text{r}^\text{BB}\mathbf{\tilde{\bar{H}}}_1\mathbf{\bar{F}}_\text{t}^\text{BB}-(\mathbf{\tilde{\bar{H}}}_2\mathbf{\bar{G}}_\text{r}^\text{BB}\mathbf{\tilde{\bar{H}}}_1\mathbf{\bar{F}}_\text{t}^\text{BB})^H\mathbf{\bar{W}}_\text{d}^\text{BB}+\mathbf{I}_{\text{N}_\text{s}},
\end{split}
\end{equation}
where 
\[
\mathbf{A}:=\mathbf{\tilde{\bar{H}}}_2\mathbf{\bar{G}}_\text{r}^\text{BB}\mathbf{A}_1(\mathbf{\tilde{\bar{H}}}_2\mathbf{\bar{G}}_\text{r}^\text{BB})^H+\Tr(\mathbf{\bar{G}}_\text{r}^\text{BB}\mathbf{A}_1(\mathbf{\bar{G}}_\text{r}^\text{BB})^H\tilde{\Theta}_2^H)\tilde{\Phi}_2,
\]
\[
\mathbf{A}_1:=\mathbf{\tilde{\bar{H}}}_1\mathbf{\bar{F}}_\text{t}^\text{BB}(\mathbf{\tilde{\bar{H}}}_1\mathbf{\bar{F}}_\text{t}^\text{BB})^H+\Tr(\mathbf{\bar{F}}_\text{t}^\text{BB}(\mathbf{\bar{F}}_\text{t}^\text{BB})^H\tilde{\Theta}_1^H)\tilde{\Phi}_1.
\]

Based on (\ref{EMSE}), we can derive the $\mathbf{\bar{W}}_\text{d}^\text{MMSE}$, which minimizes $\E_{\Delta_1,\Delta_2}[\mathbf{E}_\text{MSE}]$, as
\begin{equation}
\label{robwdmmse}
\mathbf{\bar{W}}_\text{d}^\text{MMSE}=(\mathbf{\tilde{\bar{H}}}_2\mathbf{\bar{G}}_\text{r}^\text{BB}\mathbf{\tilde{\bar{H}}}_1\mathbf{\bar{F}}_\text{t}^\text{BB})^H(\mathbf{A}+\mathbf{R}_{\tilde{\bar{\text{n}}}})^{-1}.
\end{equation}

Substituting (\ref{robwdmmse}) into (\ref{EMSE}), we have
\begin{equation}
\label{robmmse}
\begin{split}
&\E_{\Delta_1,\Delta_2}[\mathbf{E}_\text{MMSE}] \\
&= \mathbf{\bar{W}}_\text{d}^\text{MMSE}(\mathbf{A}+\mathbf{R}_{\tilde{\bar{\text{n}}}}) (\mathbf{\bar{W}}_\text{d}^\text{MMSE})^H - (\mathbf{\bar{W}}_\text{d}^\text{MMSE})^H\mathbf{\tilde{\bar{H}}}_2\mathbf{\bar{G}}_\text{r}^\text{BB}\mathbf{\tilde{\bar{H}}}_1\mathbf{\bar{F}}_\text{t}^\text{BB}-(\mathbf{\tilde{\bar{H}}}_2\mathbf{\bar{G}}_\text{r}^\text{BB}\mathbf{\tilde{\bar{H}}}_1\mathbf{\bar{F}}_\text{t}^\text{BB})^H\mathbf{\bar{W}}_\text{d}^\text{BB}+\mathbf{I}_{\text{N}_\text{s}}\\
&=\mathbf{I}_{\text{N}_\text{s}}-(\mathbf{\tilde{\bar{H}}}_2\mathbf{\bar{G}}_\text{r}^\text{BB}\mathbf{\tilde{\bar{H}}}_1\mathbf{\bar{F}}_\text{t}^\text{BB})^H(\mathbf{A}+\mathbf{R}_{\tilde{\bar{\text{n}}}})^{-1}\mathbf{\tilde{\bar{H}}}_2\mathbf{\bar{G}}_\text{r}^\text{BB}\mathbf{\tilde{\bar{H}}}_1\mathbf{\bar{F}}_\text{t}^\text{BB}.
\end{split}
\end{equation}

Based on (\ref{robmmse}), we can amend our results based on the imperfect channel model, using the same procedure as Section III. For $\mathbf{\bar{G}}_\text{r}^\text{BB}$, the amended expression is 
\begin{equation}
\label{robGrBB}
\mathbf{\bar{G}}_\text{r}^\text{BB} = (\mathbf{K}_1+\lambda^r(\mathbf{\bar{F}}_\text{r}^\text{RF})^H\mathbf{\bar{F}}_\text{r}^\text{RF})^{-1}(\mathbf{\tilde{\bar{H}}}_2)^H(\mathbf{\bar{W}}_\text{d}^\text{MMSE})^H\mathbf{\bar{V}}(\mathbf{\bar{F}}_\text{t}^\text{BB})^H(\mathbf{K}_2+\sigma_1^2\mathbf{\bar{W}}_\text{r}^\text{RF}(\mathbf{\bar{W}}_\text{r}^\text{RF})^H)^{-1},
\end{equation}
where 
\[
\mathbf{\bar{V}}=\frac{\E_{\Delta_1,\Delta_2}[\mathbf{E}_\text{MMSE}]^{-1}}{\log2},
\]
\[
\mathbf{K}_1 := (\mathbf{\bar{W}}_\text{d}^\text{MMSE}\mathbf{\tilde{\bar{H}}}_2)^H\mathbf{\bar{V}}\mathbf{\bar{W}}_\text{d}^\text{MMSE}\mathbf{\tilde{\bar{H}}}_2+\Tr(\tilde{\Phi}_2(\mathbf{\bar{W}}_\text{d}^\text{MMSE})^H\mathbf{\bar{V}}\mathbf{\bar{W}}_\text{d}^\text{MMSE})\tilde{\Theta}_2,
\]
\[
\mathbf{K}_2 := \mathbf{\tilde{\bar{H}}}_1\mathbf{\bar{F}}_\text{t}^\text{BB}( \mathbf{\tilde{\bar{H}}}_1\mathbf{\bar{F}}_\text{t}^\text{BB})^H+\Tr(\mathbf{\bar{F}}_\text{t}^\text{BB}(\mathbf{\bar{F}}_\text{t}^\text{BB})^H\tilde{\Theta}_1)\tilde{\Phi}_1.
\]

For $\bar{\mathbf{F}}_\text{t}^\text{BB}$, the amended expression is 
\begin{equation}
\label{robFtBB}
\bar{\mathbf{F}}_\text{t}^\text{BB}=(\mathbf{T}_1+\lambda^t_1(\mathbf{\bar{F}}_\text{t}^\text{RF})^H\mathbf{\bar{F}}_\text{t}^\text{RF}+ \lambda^t_2\mathbf{T}_2)^{-1}(\mathbf{\bar{V}}\mathbf{\bar{W}}_\text{d}^\text{MMSE}\tilde{\bar{\mathbf{H}}}_2\mathbf{\bar{G}}_\text{r}^\text{BB}\tilde{\mathbf{\bar{H}}}_1)^H,
\end{equation}
where
\[
\mathbf{T}_1 :=\mathbf{\tilde{\bar{H}}}_1^H(\mathbf{\bar{G}}_\text{r}^\text{BB})^H\mathbf{B}\mathbf{\bar{G}}_\text{r}^\text{BB}\mathbf{\tilde{\bar{H}}}_1+\Tr(\tilde{\Phi}_1(\mathbf{\bar{G}}_\text{r}^\text{BB})^H\mathbf{B}\mathbf{\bar{G}}_\text{r}^\text{BB})\tilde{\Theta}_1,
\]
\[
\mathbf{T}_2 :=  (\mathbf{F}_\text{r}^\text{RF}\mathbf{G}_\text{r}^\text{BB}\tilde{\mathbf{H}}_1)^H\mathbf{F}_\text{r}^\text{RF}\mathbf{G}_\text{r}^\text{BB}\tilde{\mathbf{H}}_1 + \Tr(\tilde{\Phi}_1(\mathbf{F}_\text{r}^\text{RF}\mathbf{G}_\text{r}^\text{BB})^H\mathbf{F}_\text{r}^\text{RF}\mathbf{G}_\text{r}^\text{BB})\tilde{\Theta}_1,
\]
\[
\mathbf{B}:=\mathbf{\tilde{\bar{H}}}_2^H(\mathbf{\bar{W}}_\text{d}^\text{MMSE})^H\mathbf{V}\mathbf{\bar{W}}_\text{d}^\text{MMSE}\mathbf{\tilde{\bar{H}}}_2+\Tr(\tilde{\Phi}_2(\mathbf{\bar{W}}_\text{d}^\text{MMSE})^H\mathbf{V}\mathbf{\bar{W}}_\text{d}^\text{MMSE})\tilde{\Theta}_2.
\]

Based on the above modifications, our robust baseband design for $\mathbf{\bar{G}}_\text{r}^\text{BB}$ and $\mathbf{\bar{F}}_\text{t}^\text{BB}$ is as follows:
\begin{enumerate}
\item Calculate the MMSE receiver {$\mathbf{\bar{W}}_\text{d}^\text{MMSE}$} in Eq. (\ref{robwdmmse}) and the MMSE matrix {$\E_{\Delta_1,\Delta_2}[\mathbf{E}_\text{MMSE}]$} in Eq. (\ref{robmmse}).
\item Update {$\mathbf{\bar{V}}$} by setting  {$\mathbf{\bar{V}}=\frac{\E_{\Delta_1,\Delta_2}[\mathbf{E}_\text{MMSE}]^{-1}}{\log2}$}.
\item Fix {$\mathbf{\bar{V}}$} and {$\mathbf{\bar{F}}_\text{t}^\text{BB}$}, then we find {$\mathbf{\bar{G}}_\text{r}^\text{BB}$} that minimizes {$\Tr(\mathbf{\bar{V}}\E_{\Delta_1,\Delta_2}[\mathbf{E}_\text{MMSE}])$} under the power constraints. The solution is given by Eq.(\ref{robGrBB}).
\item Fix {$\mathbf{\bar{V}}$} and {$\mathbf{\bar{G}}_\text{r}^\text{BB}$}, then we find {$\mathbf{\bar{F}}_\text{t}^\text{BB}$} that minimizes {$\Tr(\mathbf{\bar{V}}\E_{\Delta_1,\Delta_2}[\mathbf{E}_\text{MMSE}])$} under the power constraints. The solution is given by Eq.(\ref{robFtBB}).
\end{enumerate}
After we obtain $\mathbf{\bar{G}}_\text{r}^\text{BB}$ and $\mathbf{\bar{F}}_\text{t}^\text{BB}$, we will use MMSE-SIC to design $\mathbf{\bar{W}}_\text{d}^\text{BB}$, which is the same as what we did in Section III.

\section{Simulation Results}
\subsection{non-robust case}
In this section, we consider a relay MIMO system consisting of one source node equipped with a $\text{N}_\text{t} = 64$ antenna array, a relay node with an $\text{N}_\text{r} = 32$ antenna array and a destination node with a $\text{N}_\text{d}=48$ antenna array unless other number of antennas are specifically mentioned. The number of antennas is chosen from \cite{xue2018relay} for the purpose of the comparison. For simplicity, we use the channel model in Eq. (\ref{5.1}) for channel realization. Due to the limited scattering characteristic of the mmWave channels, the number of paths should be less than the number of relay antennas. Here, we assume each channel has $\text{L}=20$ paths. The $\varphi_l$ of each path is assumed to be uniformly distributed in $[0, 2\pi]$. The results are averaged over $2000$ channel realizations. The $\text{SNR}$ of the source-to-relay link and the relay-to-destination are assumed to be the same. In the simulation, we calculate the variances of AWGN noises $\sigma_1$ and $\sigma_2$ according to the source power and the relay power to maintain the same $\text{SNR}$.
\begin{figure}[htbp]
\centering
\begin{minipage}{.4\textwidth}
\vspace*{-0.2cm}
\includegraphics[scale = 0.4]{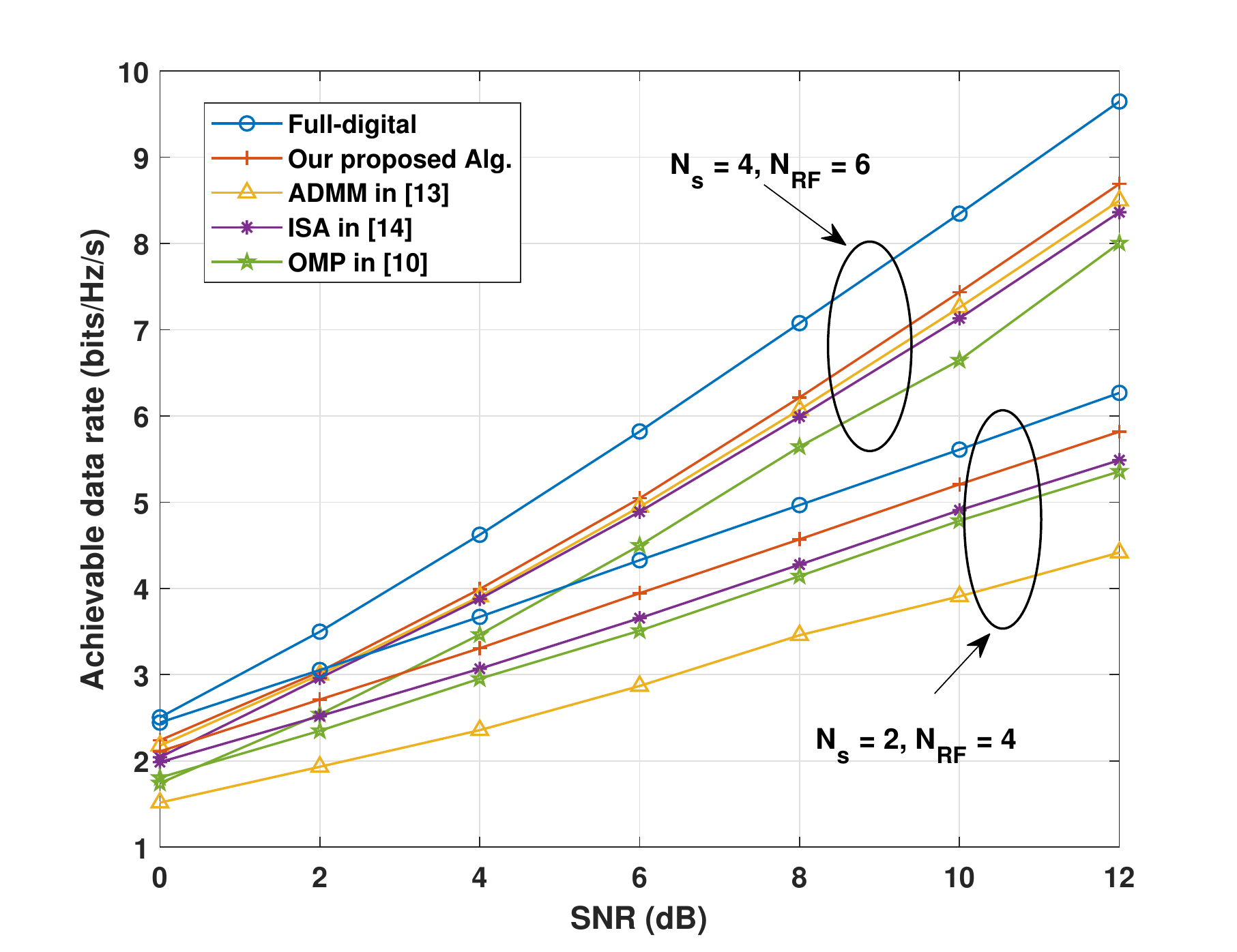}
\end{minipage}
\hspace{-0.1 cm}
\begin{minipage}{.4\textwidth}
\vspace*{-0.2cm}
\includegraphics[scale = 0.4]{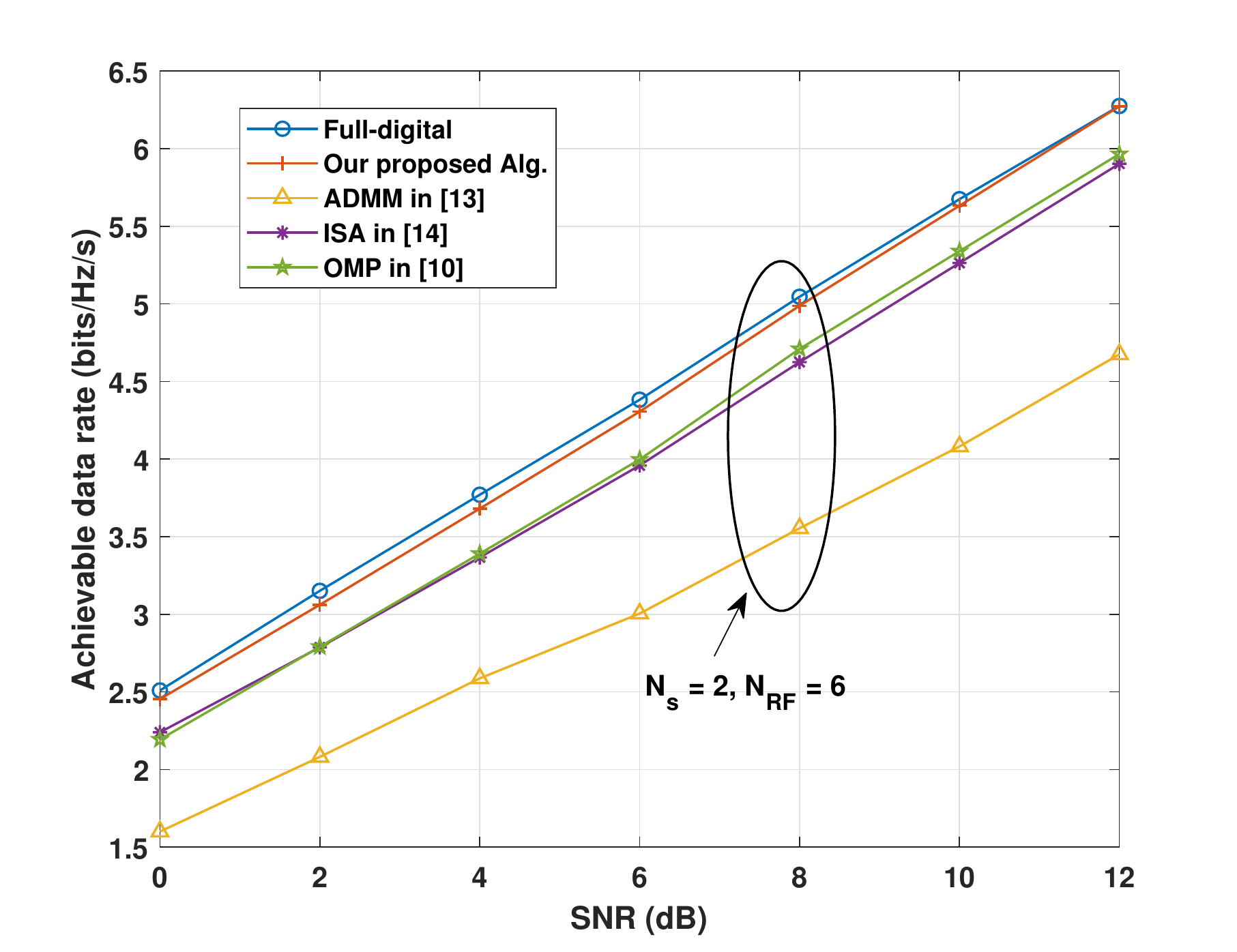}
\end{minipage}
\caption{ Achievable rate comparison with $64\times 32\times 48$ when $E_s=E_r=\text{N}_\text{s}$ }
\label{rate1}
\end{figure}
\begin{figure}[htbp]
\centering
\begin{minipage}{.4\textwidth}
\vspace*{-0.2cm}
\includegraphics[scale = 0.4]{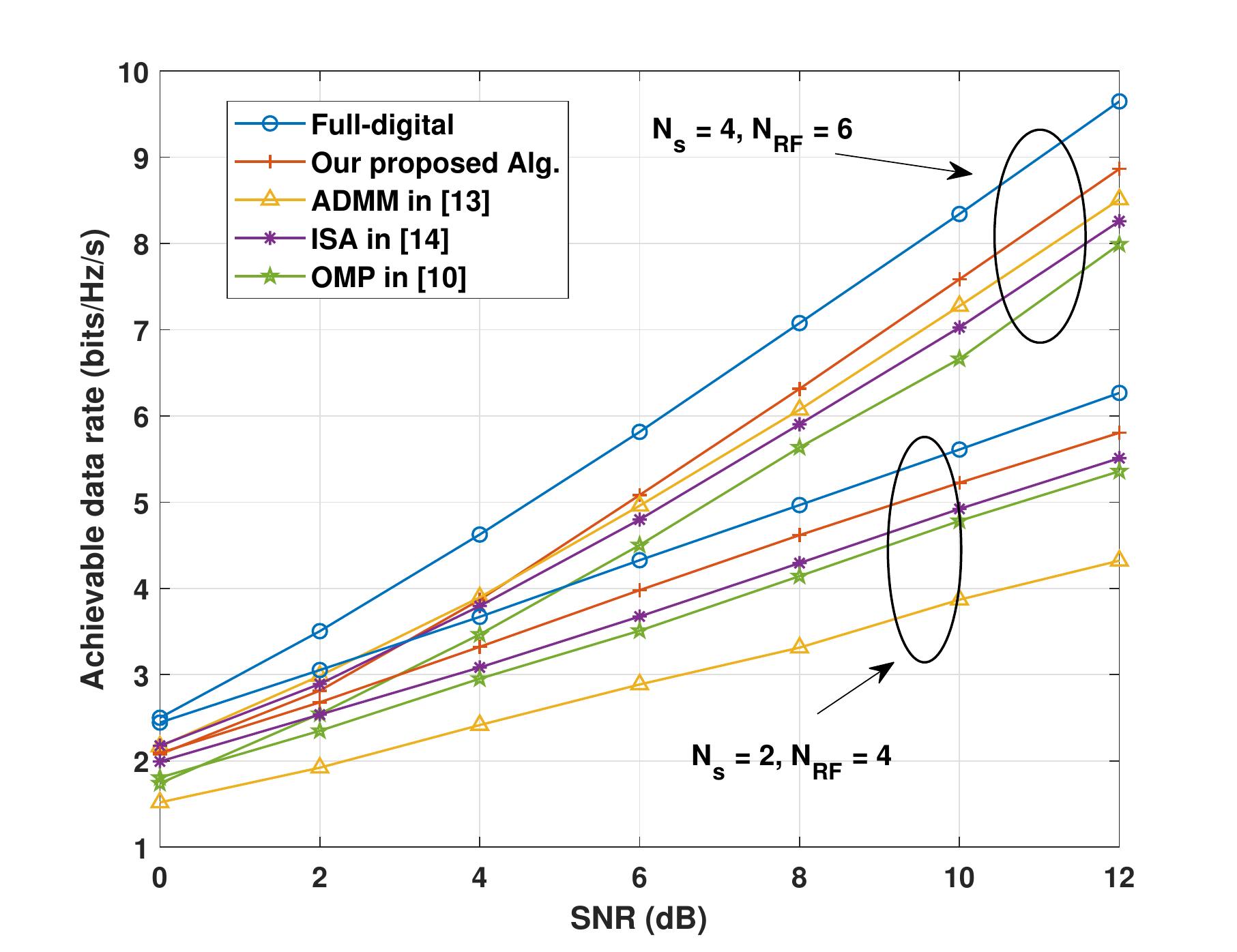}
\end{minipage}
\hspace{-0.1 cm}
\begin{minipage}{.4\textwidth}
\vspace*{-0.2cm}
\includegraphics[scale = 0.4]{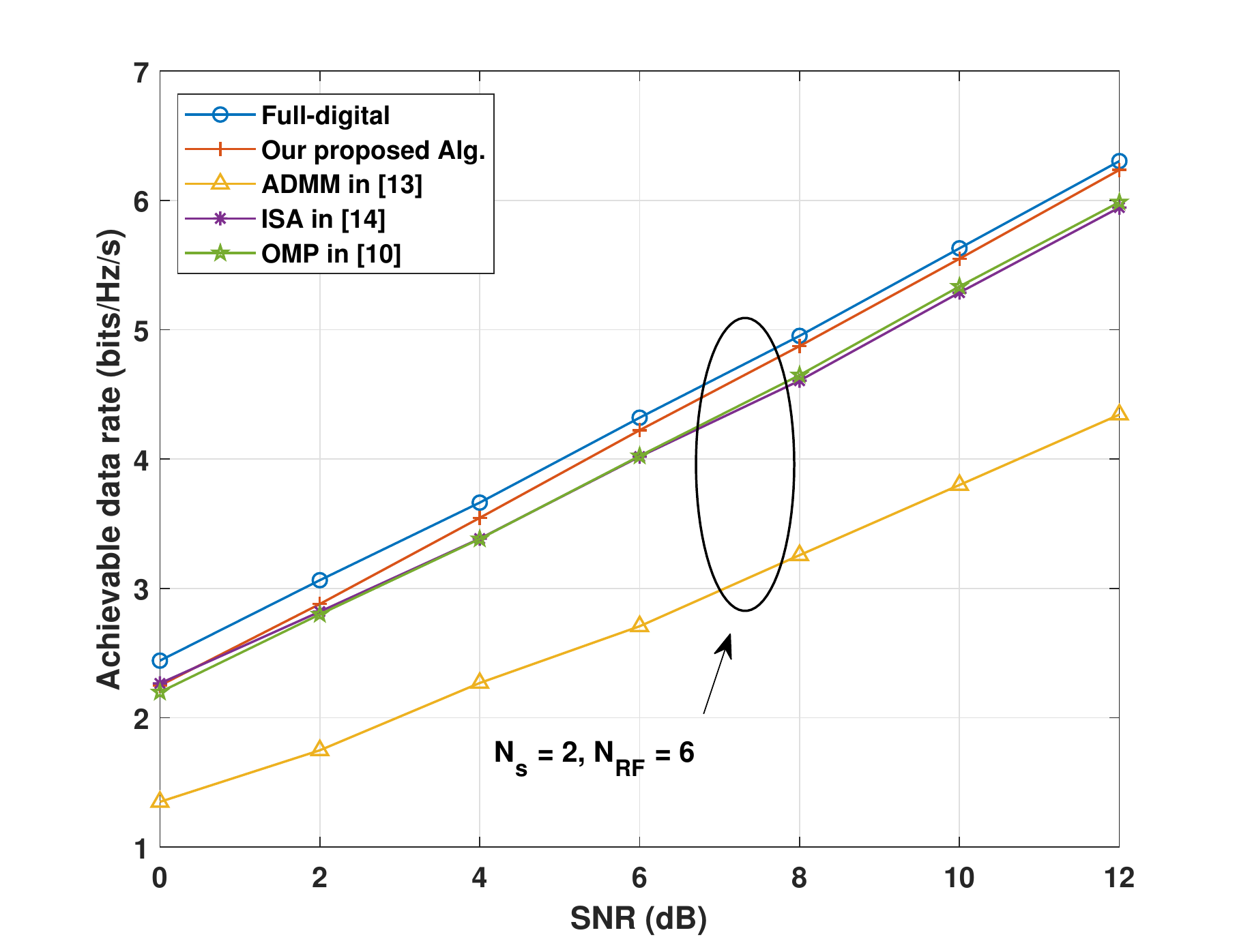}
\end{minipage}
\caption{ Achievable rate comparison with $64\times 32\times 48$ when $E_s = 2E_r=2\text{N}_\text{s}$}
\label{rate_P}
\end{figure}

In Fig. \ref{rate1}, we equally set the power of source node and the relay node, all to be $\text{N}_\text{s}$. We compare our algorithm with the ADMM in \cite{tsinos2018hybrid}, the ISA in \cite{xue2018relay} and the OMP in \cite{single_relay_mmwave} in terms of the achievable data rate. We use three scenarios: i) the number of data streams is $\text{N}_\text{s} = 4$ and the number of RF chains is $\text{N}_\text{RF} = 6$; ii) the number of data streams is $\text{N}_\text{s} = 2$ and the number of RF chains is $\text{N}_\text{RF} = 4$; iii) the number of data streams is $\text{N}_\text{s} = 2$ and the number of RF chains is $\text{N}_\text{RF} = 6$. The full-digital method is used as a benchmark, where we use the singular matrices of $\mathbf{H}_1$ and $\mathbf{H}_2$ as the precoding/combing matrices. When $\text{N}_\text{s}=4$, our algorithm outperforms ADMM by 2$\%$, ISA by 4$\%$ and OMP by 9$\%$ at $\text{SNR}=12~\text{dB}$. When $\text{N}_\text{s}=2$ and $\text{N}_\text{RF}=4$, our algorithm outperforms ADMM by 32$\%$, ISA by 6$\%$ and OMP by 9$\%$ at $\text{SNR}=12~\text{dB}$. When $\text{N}_\text{s}=2$ and $\text{N}_\text{RF}=6$, our algorithm outperforms ADMM by 34$\%$, ISA by 6 $\%$ and OMP by 5$\%$ at $\text{SNR}=12~\text{dB}$.

In Fig. \ref{rate_P}, we set $\text{E}_\text{s}=2\text{E}_\text{r}=2\text{N}_\text{s}$. Our proposed algorithm outperforms the other three methods in three scenarios.  When $\text{N}_\text{s}=4$, our algorithm can provide 4$\%$, 7$\%$ and 11$\%$ gains over ADMM, ISA and OMP, respectively, at $SNR=12~dB$. When $\text{N}_\text{s}=2$ and $\text{N}_\text{RF}=4$, our algorithm can provide a gain of 34$\%$ over ADMM, 5$\%$ over ISA and 8$\%$ over OMP at $\text{SNR}=12~\text{dB}$. When $\text{N}_\text{s}=2$ and $\text{N}_\text{RF}=6$, our algorithm can provide a gain of 44$\%$ over ADMM, 5$\%$ over ISA and 4$\%$ over OMP at $\text{SNR}=12~\text{dB}$.

\begin{figure}[htbp]
\centering
\includegraphics[width=10cm]{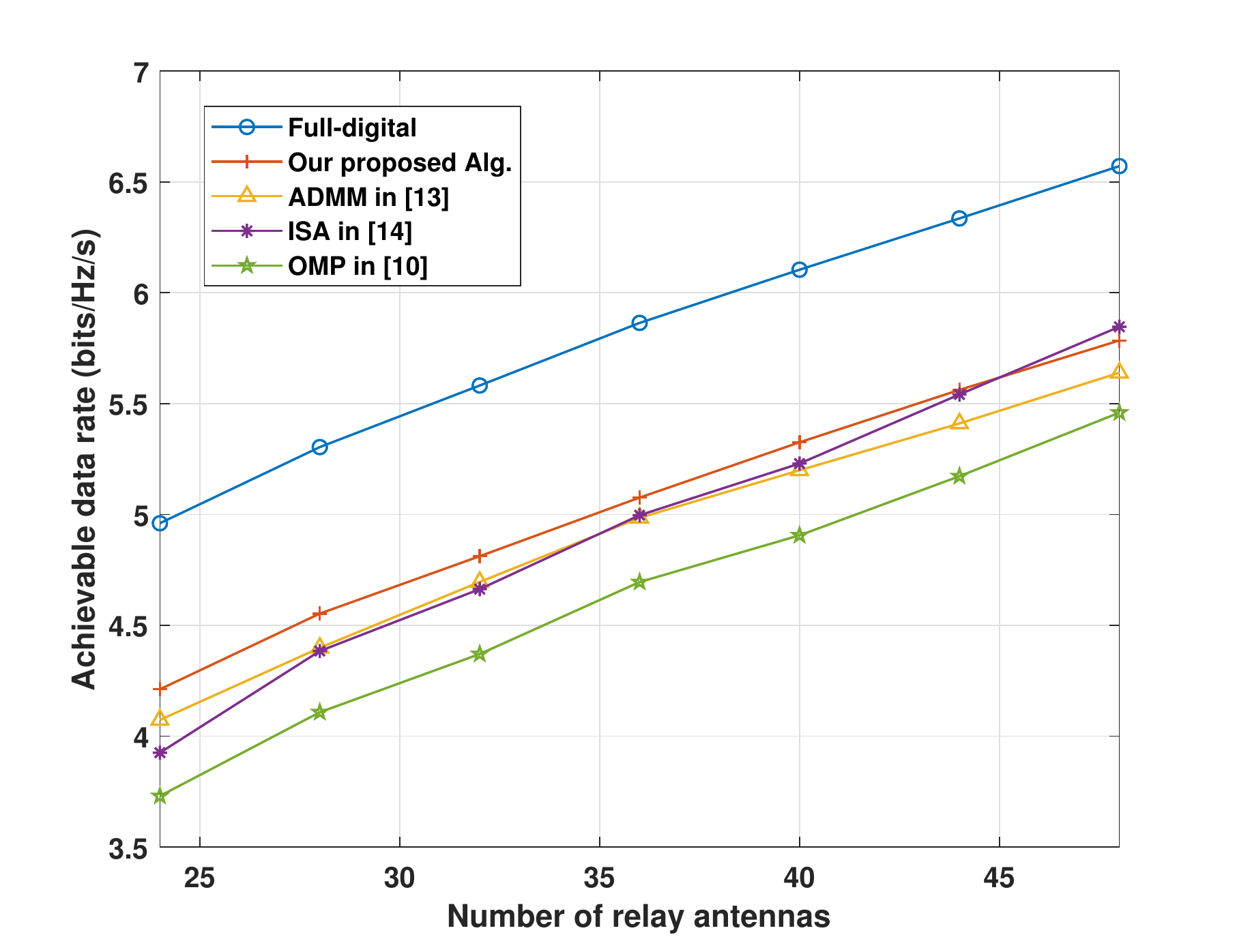}
\caption{ Achievable rate comparison with different relay antennas when $\text{N}_\text{s}=4$, $\text{N}_\text{RF}=6$ and $\text{SNR} = 5 \text{dB}$}
\label{rate2}
\vspace*{-1cm}
\end{figure}
\begin{figure}[htbp]
\centering
\includegraphics[width=10cm]{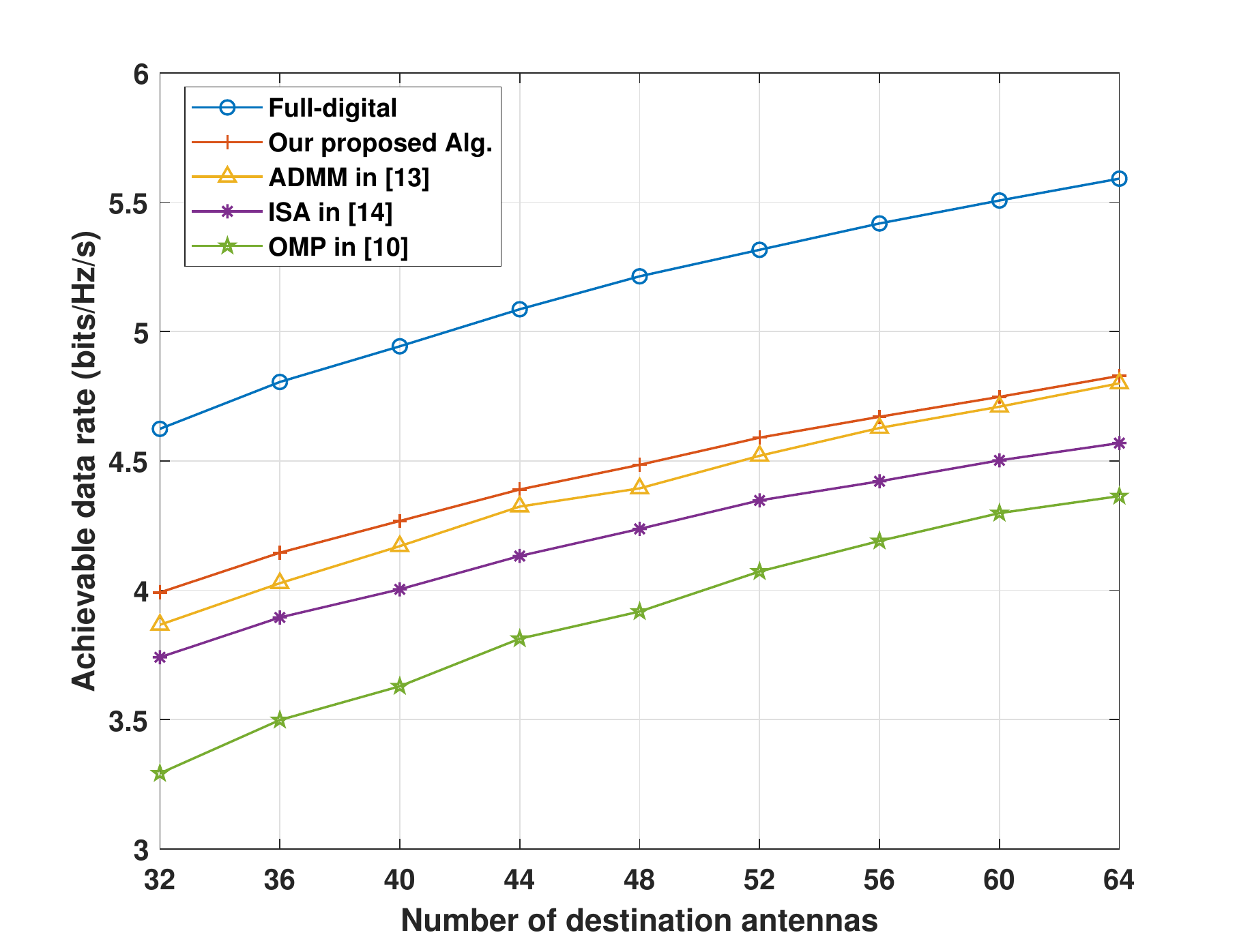}
\caption{ Achievable rate comparison with different destination antennas when $\text{N}_\text{s}=4$, $\text{N}_\text{RF}=6$ and $\text{SNR} = 5 ~\text{dB}$}
\label{rate_d}
\end{figure}

Fig. \ref{rate2} compares the achievable rate of different algorithms for different number of relay antennas when $\text{N}_\text{s}=4$, $\text{N}_\text{RF}=6$ and $\text{SNR} = 5~\text{dB}$.  The full-digital method is used as a benchmark. As expected, when the number of antennas at the relay node increases, the performance of all different algorithms improves because of the additional antenna gain. Our proposed method has the best achievable rate performance among the four methods except for $\text{N}_\text{r} = 48$. When $\text{N}_\text{r}=48$, ISA has the highest achievable rate among the four methods. However, as the number of antennas at the relay node increases, the complexity of the ISA increases greatly, which will lead to a high power consumption. 

Fig. \ref{rate_d} compares the achievable rate for different number of antennas at the destination node when $\text{N}_\text{s}=4$, $\text{N}_\text{RF}=6$ and $\text{SNR} = 5~\text{dB}$.  Similar to Fig. \ref{rate2}, when the number of antennas at the destination node increases, the performance of all different algorithms improves because of the additional antenna gain. Our proposed method has the best achievable rate performance among the four methods.
\begin{figure}[htbp]
\vspace*{-2cm}
\centering
\includegraphics[width=10cm]{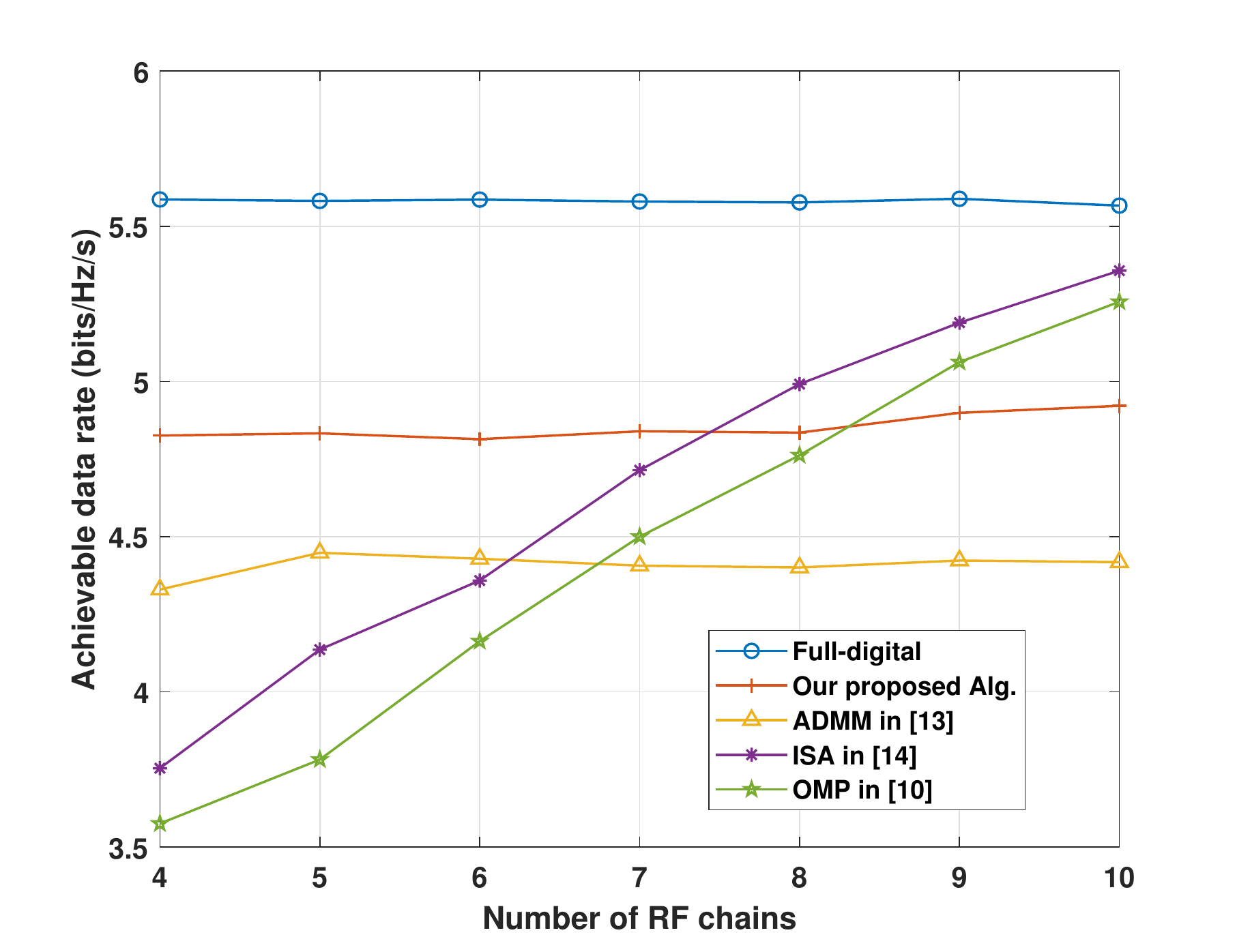}
\caption{ Achievable rate comparison with different RF chains when $\text{N}_\text{s}=4$ and $\text{SNR} = 5 ~\text{dB}$ using channel model (\ref{5.1})}
\label{rate3}
\end{figure}
\begin{figure}[htbp]
\vspace*{-1cm}
\centering
\includegraphics[width=10cm]{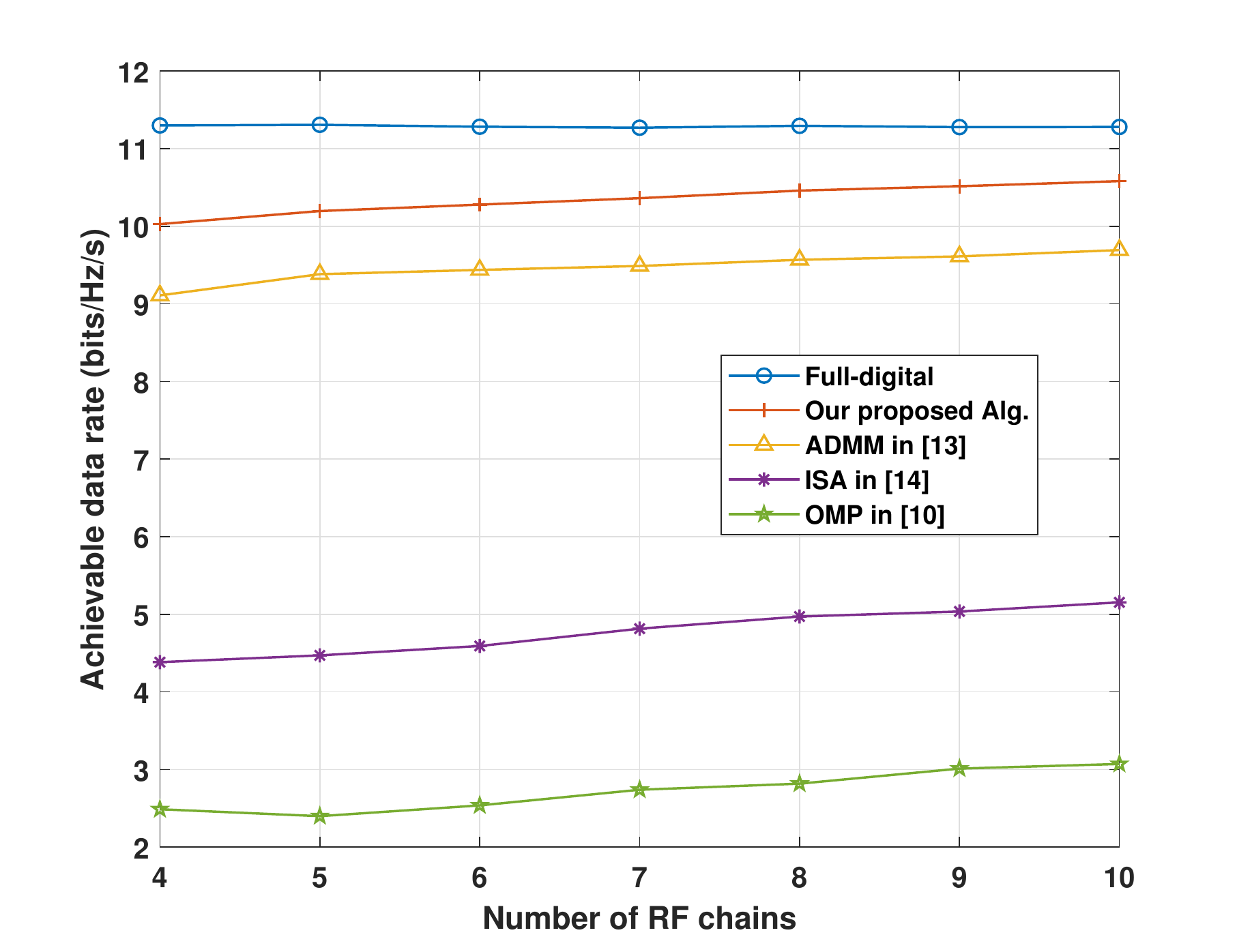}
\caption{ Achievable rate comparison with different RF chains when $\text{N}_\text{s}=4$ and $SNR = 5 dB$ using channel model (\ref{5})}
\label{rate4}
\vspace*{-1cm}
\end{figure}

Fig. \ref{rate3} compares the achievable rate among the four methods for different number of RF chains when $\text{N}_\text{s}=4$ and $\text{SNR} = 5~\text{dB}$. Since our proposed method is designed to maximize the mutual information between the destination node and the source node after RF precoding/combining, the gap between our method and the full-digital method is more-or-less fixed, which is caused by the analog processing. However, ISA and OMP are approximation algorithms jointly iterating between the RF and the baseband. Therefore, as the number of RF chains increases, the performance improves. When the number of RF chains is larger than 8,  ISA and OMP will outperform our proposed algorithm. However, larger number of RF chains leads to higher complexity and more power consumption. Also, the performance of the approximation algorithms depends on the limited scattering characteristic of the channel. The more sparse the channel is, the better performance the approximation algorithms achieve. In Fig. \ref{rate3}, we use the highly limited scattering channel model in  (\ref{5.1}), where each scatter only contributes to one path, thus the approximation algorithms have good performance. If we use the general channel model in (\ref{5}), the performance of approximation algorithms degrades greatly as shown in Fig. \ref{rate4}. In  Fig. \ref{rate4}, we set the number of propagation paths $\text{N}_\text{cl}$ in each scatter to be $2$ and the number of scatters $L$ to be $20$.  In this case, the performance of ISA and OMP falls far behind our proposed algorithm.

Fig. \ref{convergence} shows the convergence performance of different algorithms with respect to the number of iterations. In our algorithm, we update the WMMSE matrix, the digital relay matrix and the digital precoding matrix sequentially in each iteration. In ISA,  the digital relay filter, the analog relay receiver and the analog relay precoder are updated sequentially in each iteration. In ADMM, the source node, the relay node and the destination node are optimized alternatively in each iteration. In Fig. \ref{convergence}, our algorithm has the fastest convergence rate while ADMM has the slowest convergence rate. Moreover, our algorithm has much lower complexity in each iteration compared with ISA. ISA needs to solve three optimization sub-problems, and in each sub-problem it needs to solve an optimization problem through an iterative method. In our algorithm, we have closed-form solutions for each step. In addition, since we preform the baseband processing after the RF processing, the matrix dimensions are greatly reduced compared to ISA. 

\begin{figure}[htbp]
\centering
\includegraphics[width=10cm]{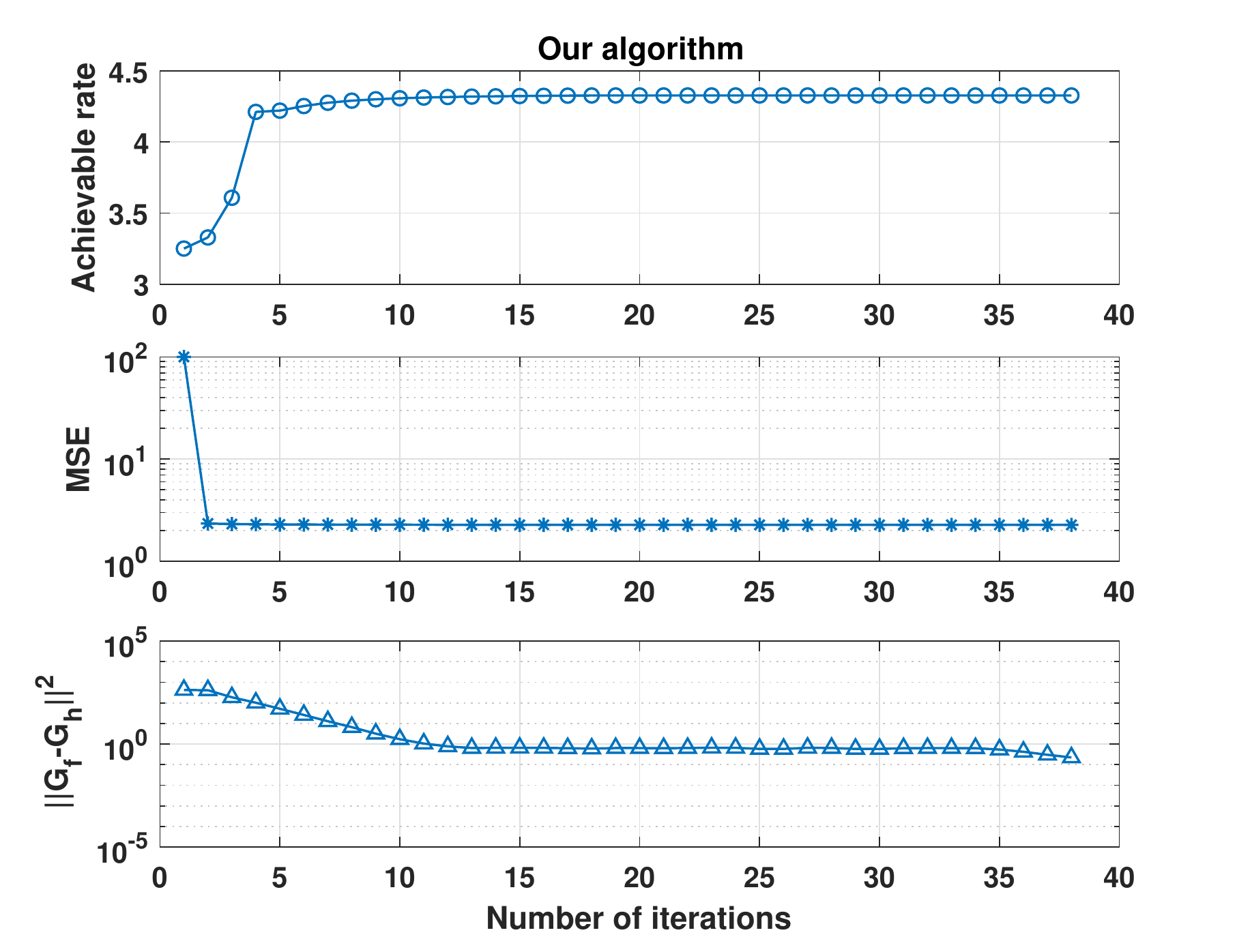}
\caption{Convergence rate comparison}
\label{convergence}
\end{figure}
\begin{figure}[htbp]
\centering
\includegraphics[width=10cm]{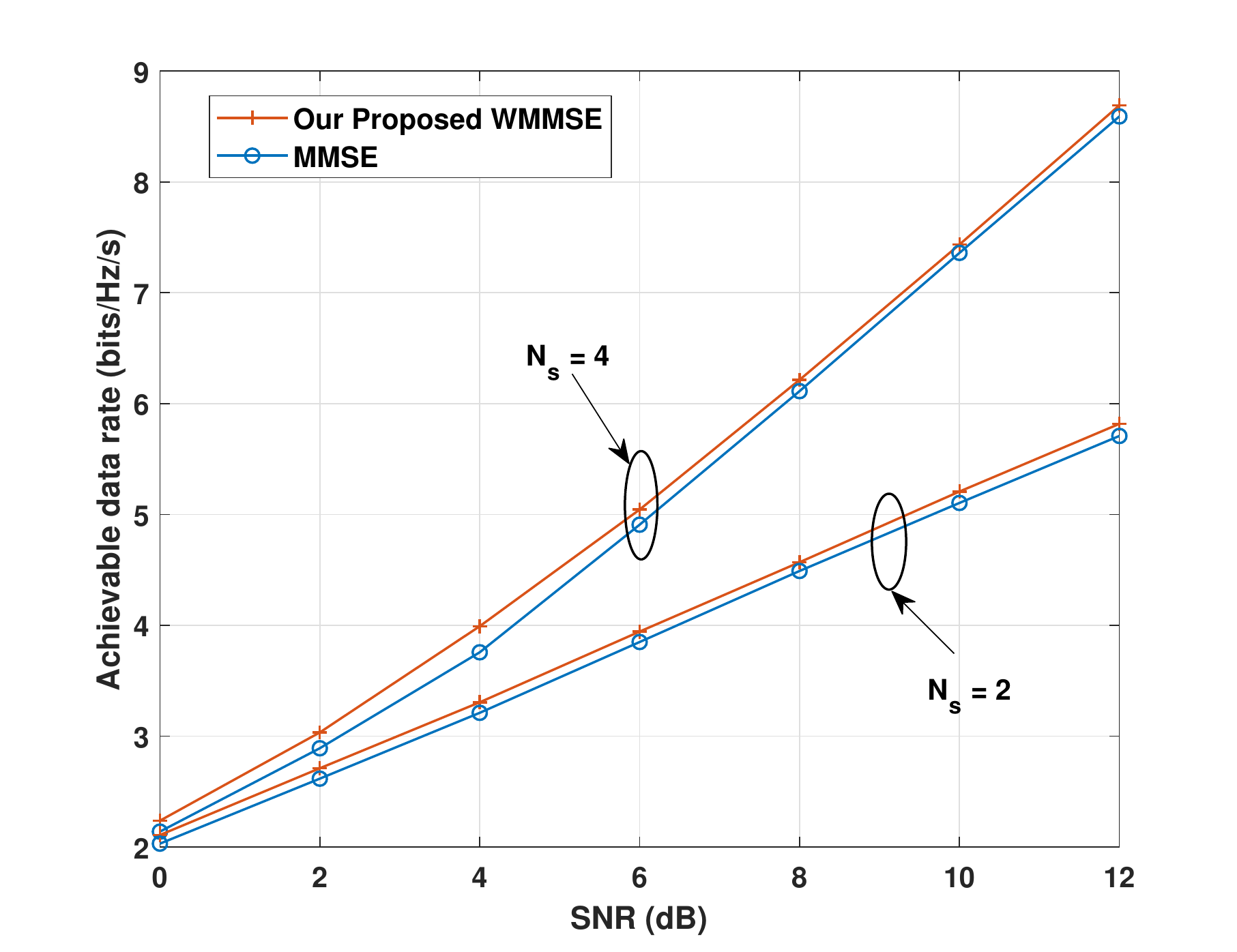}
\caption{Baseband algorithms comparison}
\label{bb_comp}
\end{figure}
Fig. \ref{bb_comp} compares the baseband processing algorithms. Note that we apply the MMSE algorithm only on the baseband, i.e., on the $\mathbf{\tilde{H}}_1$ and $\mathbf{\tilde{H}}_2$. Our proposed WMMSE algorithm outperforms the MMSE algorithm in terms of the achievable data rate since we optimized the mutual information $I(\bf{s},\bf{y}_d)$. In fact, if we set our weight matrix to be the identity matrix, our algorithm degenerates to the MMSE algorithm. Therefore, the MMSE algorithm can be considered as a special case of our proposed WMMSE algorithm and our algorithm strictly performs better than MMSE.

\subsection{Robust case}

\begin{figure}[htbp]
\centering
\includegraphics[width=10cm]{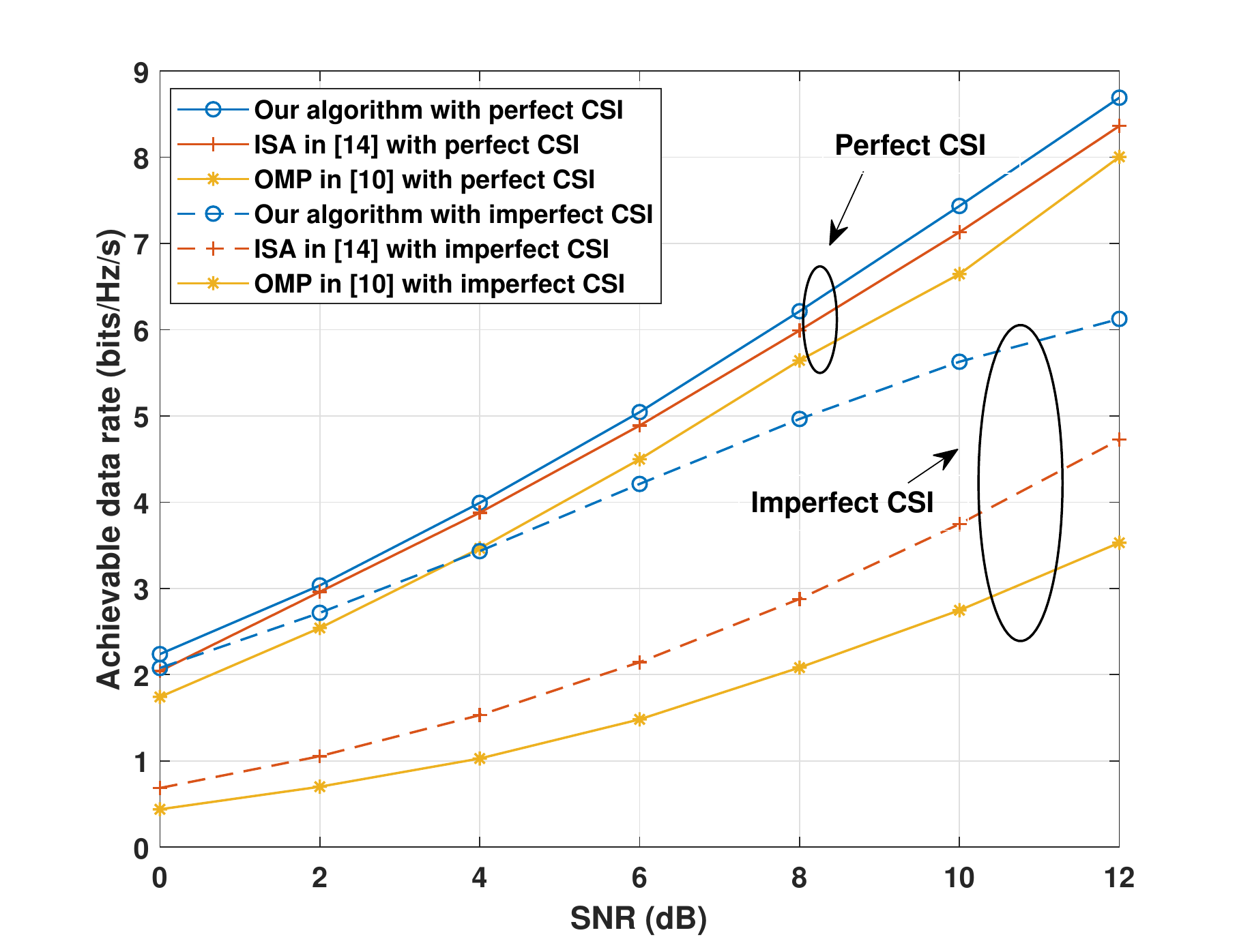}
\caption{The effect of imperfect CSI when the error covariance $\sigma_e^2 = 0.1$, $\alpha=0.6$ and $\beta=0.4$}
\label{imperfect_CSI}
\end{figure}
\begin{figure}[htbp]
\centering
\includegraphics[width=10cm]{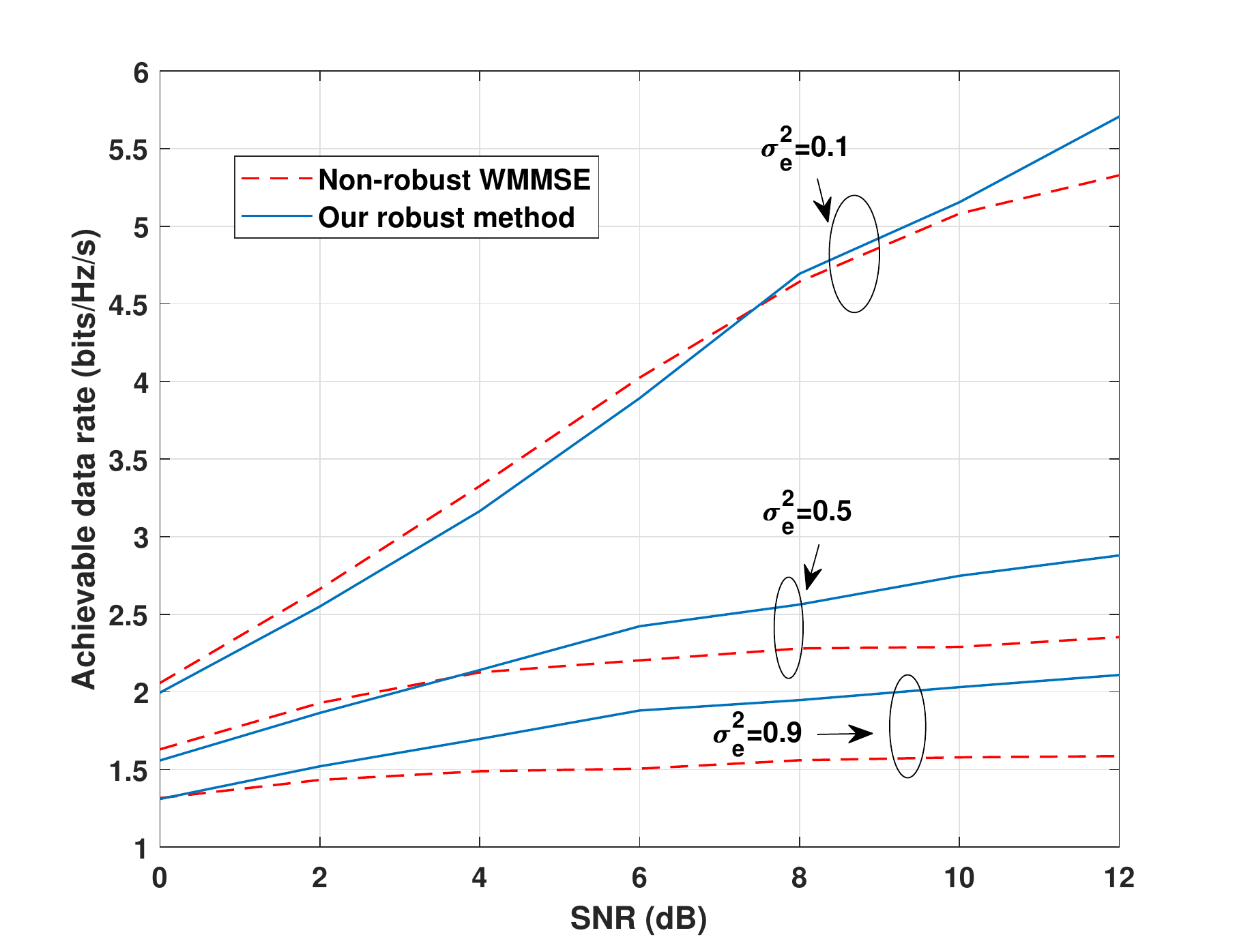}
\caption{Achievable rate comparison when $\alpha=0$, $\beta=0$}
\label{rob1}
\end{figure}
\begin{figure}[htbp]
\centering
\includegraphics[width=10cm]{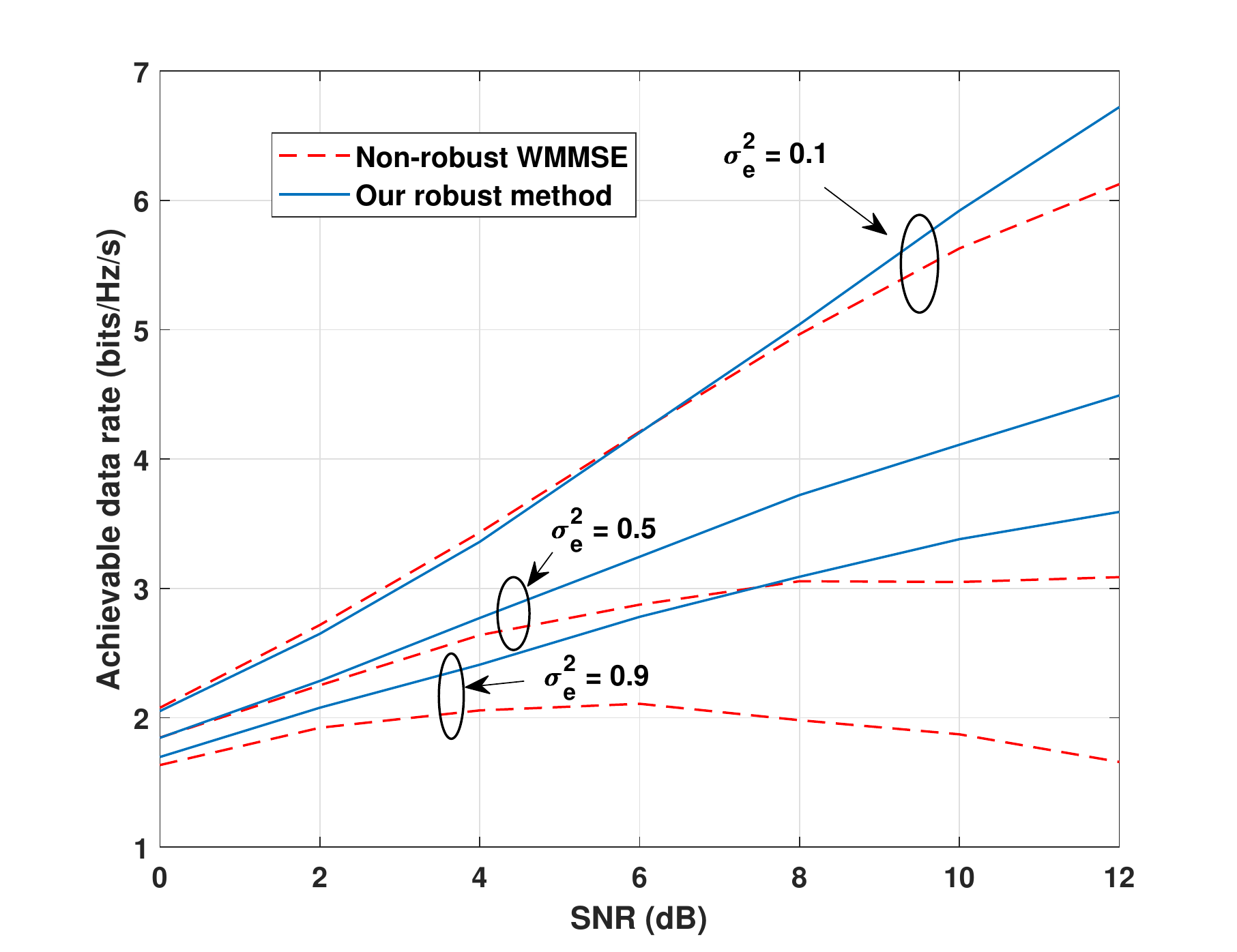}
\caption{Achievable rate comparison when $\alpha=0.5$, $\beta=0.5$}
\label{rob2}
\end{figure}
\begin{figure}[htbp]
\centering
\includegraphics[width=10cm]{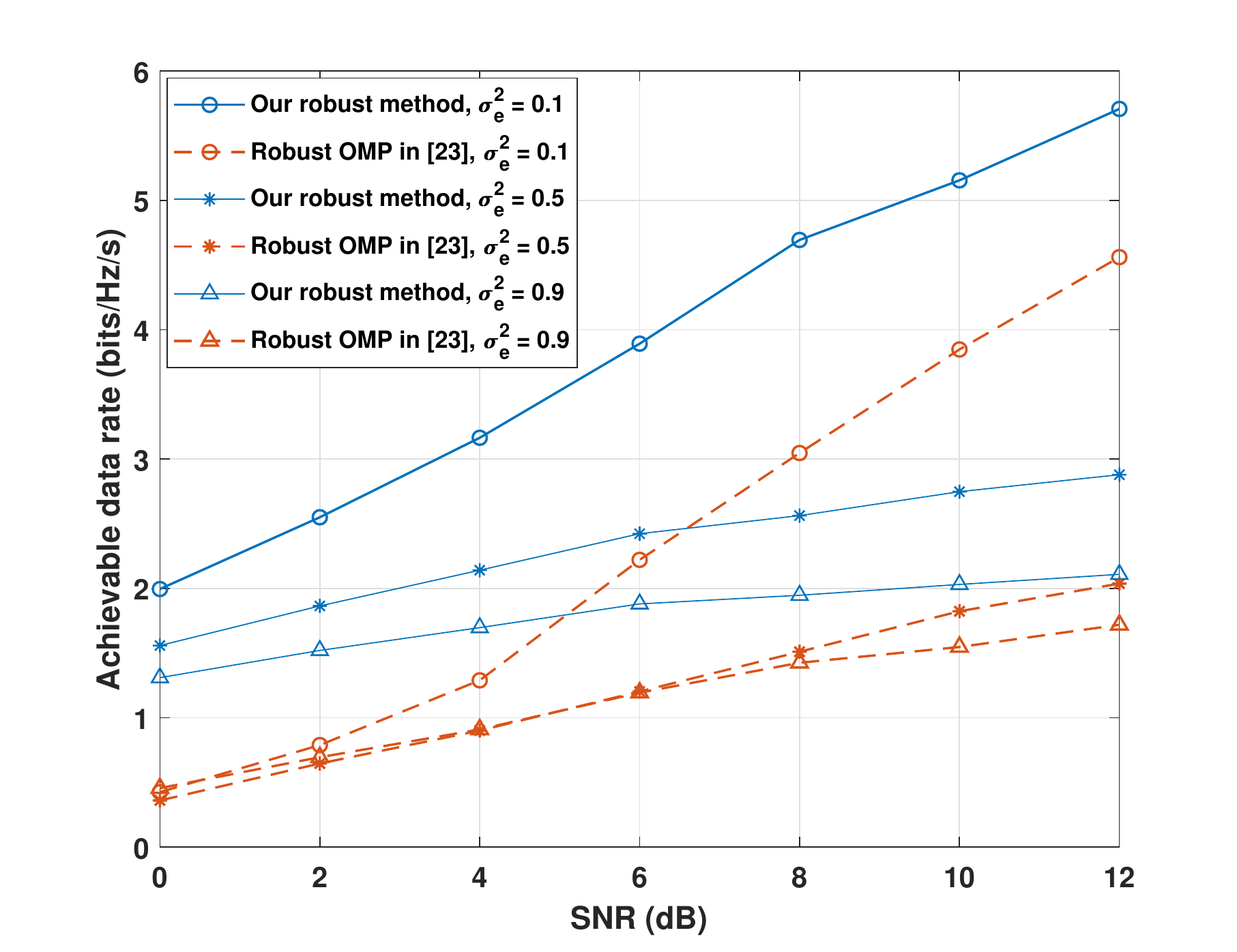}
\caption{Achievable rate comparison with the OMP algorithm in \cite{luo2018robust} when $\alpha=0$, $\beta=0$}
\label{rob_com}
\end{figure}

As we described in Section \uppercase\expandafter{\romannumeral4}, we adopt the channel estimation error model from \cite{zhang2008statistically,rong2011robust, xing2010robust} where the entries of the correlation matrices are selected as $\Phi_1(i,j)=\sigma_{e,1}^2\beta_1^{|i-j|}$, $\Theta_1(i,j) = \alpha_1^{|i-j|}$, $\Phi_2(i,j)=\sigma_{e,2}^2\beta_2^{|i-j|}$ and $\Theta_2(i,j) = \alpha_2^{|i-j|}$. Parameters $\alpha_1$, $\beta_1$, $\alpha_2$ and $\beta_2$ are the correlation coefficients and $\sigma_{e,1}^2$ and $\sigma_{e,2}^2$ denote the estimation error covariance. For simplicity, we assume $\alpha_1=\alpha_2=\alpha$, $\beta_1=\beta_2=\beta$ and $\sigma_{e,1}^2=\sigma_{e,2}^2=\sigma_e^2$. The antenna settings are the same as the non-robust part and the number of scatters is set to be $20$. The actual channels $\mathbf{H}_1$ and $\mathbf{H}_2$ are generated based on sparse channel model (\ref{5.1}) and the estimated channels are generated by $\mathbf{\bar{H}}_1=\mathbf{H}_1-\Phi_1^{\frac{1}{2}}\Delta_1\Theta_1^{\frac{1}{2}}$ and $\mathbf{\bar{H}}_2=\mathbf{H}_2-\Phi_2^{\frac{1}{2}}\Delta_2\Theta_2^{\frac{1}{2}}$. 

Fig. \ref{imperfect_CSI} shows the effects of the channel estimation error. We provide the performance of our algorithm and those of \cite{single_relay_mmwave,xue2018relay}. For this simulation, we have chosen $\sigma^2=0.1$, $\alpha=0.6$ and $\beta=0.4$. As shown in Fig. \ref{imperfect_CSI}, the imperfect channel information will result in severe performance degradation. The achievable data rate of \cite{single_relay_mmwave,xue2018relay} can be decreased to half of what it is for the perfect CSI.

The achievable data rate performances of the proposed robust scheme with various antenna covariance values are depicted in Figs. \ref{rob1} and \ref{rob2}. When $SNR$ is low, the estimation error can be neglected compared to the noise, therefore the non-robust algorithm achieves good performance which can be even better than that of the robust algorithm. When $SNR$ goes up, the performance of the non-robust algorithm starts to degrade. In Fig. \ref{rob2}, the performance becomes worse than that of the low $\text{SNR}$ region for large $\sigma_e^2$. Meanwhile, the proposed robust design offers significant gain considering various $\sigma_e^2$, which demonstrates the effectiveness of the modified robust transceiver optimization.

In Fig. \ref{rob_com}, we compare our robust algorithm with the OMP algorithm in \cite{luo2018robust}. We set $\alpha=0$ and $\beta=0$ for simplicity. The proposed robust design provides a large gain over the algorithm in \cite{luo2018robust} in all three $\sigma_e^2$ settings, showing the advantage of our algorithm. 

\section{Conclusion}
In this paper, we considered mmWave AF relay networks in the domain of massive MIMO. We designed the hybrid precoding/combining matrices for the source node, the relay node, and the destination node. We first performed the RF processing to decompose the channel into parallel sub-channels by compensating the phase of each eigenmode of the channel. Given the RF processing matrices, we designed the baseband matrices to maximize the mutual information.  The baseband processing is divided into two parts. We first jointly designed the source node and the relay node by making use of the equivalence between maximizing the mutual information and the WMMSE. Given the optimal baseband source and relay filters, we implemented MMSE-SIC for baseband destination node to obtain the maximal mutual information. In addition, a robust hybrid precoding/combining design was proposed for the imperfect CSI. Simulation results show that our algorithm achieves better performance with lower complexity compared with other algorithms in the literature.

\ifCLASSOPTIONcaptionsoff
  \newpage
\fi
\small
\bibliographystyle{IEEEtran}
\bibliography{IEEEabrv,refs}

\end{document}